\documentclass[11pt]{amsart}
\usepackage{amsmath}
\usepackage{esint}
\usepackage{graphicx}
\theoremstyle{plain}
\newtheorem{theorem}{Theorem}[section]
\newtheorem{lemma}[theorem]{Lemma}
\newtheorem{prop}[theorem]{Proposition}
\newtheorem{cor}[theorem]{Corollary}

\newtheorem{conjecture}[theorem]{Conjecture}

\theoremstyle{definition}

\numberwithin{equation}{section}

\newcommand{\R}{\mathbb R}
\renewcommand{\H}{\mathbb H}
\newcommand{\C}{\mathbb C}

\newcommand{\dist}{\operatorname{dist}}
\newcommand{\KN}{\text{\tiny{KN}}}
\newcommand{\tr}{\operatorname{Tr}}
\newcommand{\EM}{\text{\tiny{EM}}}
\newcommand{\MP}{\text{\tiny{MP}}}

\begin{document}

\title[The Positive mass theorem for multiple rotating charged black holes] {The Positive mass theorem for multiple rotating charged black holes}

\author[Khuri]{Marcus Khuri}
\address{Department of Mathematics\\
Stony Brook University\\
Stony Brook, NY 11794, USA}
\email{khuri@math.sunysb.edu}

\author[Weinstein]{Gilbert Weinstein}
\address{Department of Physics and Department of Mathematics\\
Ariel University\\
Ariel, 40700, Israel}
\email{gilbertw@ariel.ac.il}

\thanks{M. Khuri acknowledges the support of NSF Grant DMS-1308753.}

\begin{abstract}
In this paper a lower bound for the ADM mass is given in terms of the angular momenta and charges
of black holes present in axisymmetric initial data sets for the Einstein-Maxwell equations. This generalizes the mass-angular momentum-charge inequality obtained by Chrusciel and Costa to the case of multiple black holes. We also weaken the hypotheses used in the proof of this result for single black holes, and establish the associated rigidity statement.
\end{abstract}
\maketitle

\section{Introduction}
\label{sec1}

Based on heuristic arguments reminiscent of those used to motivate the Penrose inequality (see Appendix \ref{sec5}), one
may derive the following inequality
\begin{equation}\label{0}
m^2\geq\frac{q^{2}+  \sqrt{q^4 + 4\mathcal{J}^2}}{2},
\end{equation}
relating the ADM mass $m$, ADM angular momentum $\mathcal{J}$, and total charge of asymptotically flat axisymmetric initial data for the Einstein-Maxwell equations. This inequality implies both the mass-angular momentum inequality $m\geq\sqrt{|\mathcal{J}|}$ and the mass-charge inequality $m\geq|q|$; the later is often referred to as the positive mass theorem with charge. While the mass-charge inequality has been rigorously established in great generality \cite{GHHP}, without the axisymmetric assumption and for multiple black holes, the same is not true of the mass-angular momentum inequality or the mass-angular momentum-charge inequality \eqref{0}. For these inequalities, the axisymmetric condition is necessary as it is
related to conservation of angular momentum, without which the motivating heuristic arguments would no longer apply. In fact, counterexamples exist \cite{HuangSchoenWang} without the axisymmetric hypothesis. In this setting, and with the addition of supplementary hypotheses to be discussed below, the mass-angular momentum inequality was established for a single black hole by Dain in \cite{Dain}, and was later extended and improved upon by Schoen and Zhou \cite{SchoenZhou}. The case of multiple black holes was taken up
by Chrusciel, Li, and Weinstein \cite{ChruscielLiWeinstein} who proved the lower bound
\begin{equation}\label{0.1}
m\geq \mathcal{F}(\mathcal{J}_{1},\ldots,\mathcal{J}_{N}),
\end{equation}
where $\mathcal{F}$ is a function of the angular momentuma $\mathcal{J}_{n}$ associated with the $N$ black holes. It is an open question whether this function agrees with the predicted value $\sqrt{|\mathcal{J}|}$, where
$\mathcal{J}=\sum_{n=1}^{N}\mathcal{J}_{n}$. The inequality \eqref{0} has also been settled under certain conditions for single black holes by Chrusciel and Costa \cite{ChruscielCosta}, \cite{Costa}. It is the primary
purpose of the present article to extend this result to the case of multiple black holes, by establishing in this setting a lower bound for the mass in the spirit of \eqref{0.1}.

An initial data set $(M, g, k, E, B)$ for the Einstein-Maxwell equations consists of a 3-manifold $M$, Riemannian metric $g$, symmetric 2-tensor $k$ representing
extrinsic curvature, and vector fields $E$ and $B$ which constitute the electromagnetic field. Let $\mu_\EM$ and $J_\EM$
be the energy and momentum densities of the matter fields after contributions from the Maxwell field have been removed.
If charged matter is not present, the initial data satisfy the following set of constraints
\begin{align}\label{1}
\begin{split}
16\pi\mu_\EM &= R+(\tr_{g}k)^{2}-|k|_{g}^{2}-2(|E|_{g}^{2}+|B|_{g}^{2}),\\
8\pi J_\EM &= \operatorname{div}_{g}(k-(\tr_{g}k)g)+2E\times B,\\
\operatorname{div}_{g} E & =\operatorname{div}_{g} B=0,
\end{split}
\end{align}
where $R$ is the scalar curvature of $g$, and $(E\times B)_{i}=\epsilon_{ijl}E^{j}B^{l}$ is the cross product with
$\epsilon$ the volume form of $g$.

It will be assumed throughout that the data are axially symmetric. This means that the group of isometries of the
Riemannian manifold $(M,g)$ has a subgroup
isomorphic to $U(1)$, and that all quantities defining the initial data are invariant under the $U(1)$ action. Thus, if
$\eta$ is the Killing field associated with this symmetry, then
\begin{equation}\label{5}
\mathfrak{L}_{\eta}g=\mathfrak{L}_{\eta}k=\mathfrak{L}_{\eta}E=\mathfrak{L}_{\eta}B=0,
\end{equation}
where $\mathfrak{L}_{\eta}$ denotes Lie differentiation. We will also postulate that
$M$ has at least two ends, with one designated end being asymptotically flat, and the others being either asymptotically
flat or asymptotically cylindrical. Recall that a domain $M_{\text{end}}\subset M$ is an
asymptotically flat end if it is diffeomorphic to $\mathbb{R}^{3}\setminus\text{Ball}$, and in the coordinates given by
the asymptotic
diffeomorphism the following fall-off conditions hold
\begin{equation}\label{3}
g_{ij}=\delta_{ij}+o_{l}(r^{-\frac{1}{2}}),\text{ }\text{ }\text{ }\text{ }\partial g_{ij}\in
L^{2}(M_{\text{end}}),\text{
}\text{ }\text{ }
\text{ }k_{ij}=O_{l-1}(r^{-\lambda}),\text{ }\text{ }\text{ }\text{ }\mu_{\EM}, J_{\EM}^{i}\in L^{1}(M_{\text{end}}),
\end{equation}
\begin{equation}\label{4}
E^{i}=O_{l-1}(r^{-\lambda}),\text{ }\text{ }\text{ }\text{ }\text{ }B^{i}=O_{l-1}(r^{-\lambda}),\text{
}\text{ }\text{ }
\text{ }\lambda>\frac{3}{2},
\end{equation}
for some $l\geq 5$.\footnote{The notation $f=o_{l}(r^{-a})$ asserts that
$\lim_{r\rightarrow\infty}r^{a+j}\partial_{r}^{j}f=0$
for all $j\leq l$, and
$f=O_{l}(r^{-a})$ asserts that $r^{a+j}|\partial_{r}^{j}f|\leq C$ for all $j\leq l$. The assumption $l\geq 5$ is needed for
the results in \cite{Chrusciel} and \cite{Sokolowsky}.}

Let $M$ be simply connected.
Then it is shown in
\cite{Chrusciel} (see also \cite{Sokolowsky} for the case when cylindrical ends are present) that $M\cong\mathbb{R}^{3}\setminus\sum_{n=1}^{N}p_{n}$, and that there exists a global (cylindrical)
Brill coordinate system $(\rho,z,\phi)$ on $M$, where the points $p_{n}$ representing black holes all lie on the
$z$-axis, and in which the Killing field is given by $\eta=\partial_{\phi}$.  In these coordinates the metric takes a
simple form
\begin{equation}\label{16}
g=e^{-2U+2\alpha}(d\rho^{2}+dz^{2})+\rho^{2}e^{-2U}(d\phi+A_{\rho}d\rho+A_{z}dz)^{2},
\end{equation}
where $\rho e^{-U}(d\phi+A_{\rho}d\rho+A_{z}dz)$ is the dual 1-form to $|\eta|^{-1}\eta$ and all coefficient functions
are independent of $\phi$. Let $M_{\text{end}}^{0}$ denote the designated asymptotically flat end associated with the
limit $r=\sqrt{\rho^{2}+z^{2}}\rightarrow\infty$. Then in this end
\begin{equation}\label{17}
U=o_{l-3}(r^{-\frac{1}{2}}),\text{ }\text{ }\text{ }\text{ }\alpha=o_{l-4}(r^{-\frac{1}{2}}),\text{ }\text{ }\text{
}\text{ }A_{\rho}=\rho o_{l-3}(r^{-\frac{5}{2}}),\text{ }\text{ }\text{ }\text{ }A_{z}=o_{l-3}(r^{-\frac{3}{2}}).
\end{equation}
The remaining ends associated with the points $p_{n}$ will be denoted by $M_{\text{end}}^{n}$, and are associated with
the limit $r_{n}\rightarrow 0$, where $r_{n}$ is the Euclidean distance to $p_{n}$. As the remaining ends may be either asymptotically flat or asymptotically cylindrical, we list both types of asymptotics
\begin{equation}\label{18}
U=2\log r_{n}+o_{l-4}(r_{n}^{\frac{1}{2}}),\text{ }\text{ }\text{ }\text{ }\alpha=o_{l-4}(r_{n}^{\frac{1}{2}}),\text{
}\text{ }\text{ }\text{ }A_{\rho}=\rho o_{l-3}(r_{n}^{\frac{1}{2}}),\text{ }\text{ }\text{ }\text{
}A_{z}=o_{l-3}(r_{n}^{\frac{3}{2}}),
\end{equation}
\begin{equation}\label{19}
U=\log r_{n}+O_{l-4}(1),\text{ }\text{ }\text{ }\text{ }\alpha=O_{l-4}(1),\text{
}\text{ }\text{ }\text{ }A_{\rho}=\rho o_{l-3}(r_{n}^{\frac{1}{2}}),\text{ }\text{ }\text{ }\text{
}A_{z}=o_{l-3}(r_{n}^{\frac{3}{2}}),
\end{equation}
respectively.

The fall-off conditions in the designated asymptotically flat end guarantee that the ADM mass, ADM angular momentum, and
total charges are well-defined by the following limits
\begin{equation}\label{6}
m=\frac{1}{16\pi}\int_{S_{\infty}}(g_{ij,i}-g_{ii,j})\nu^{j},
\end{equation}
\begin{equation}\label{7}
\mathcal{J}=\frac{1}{8\pi}\int_{S_{\infty}}(k_{ij}-(\tr_{g} k)g_{ij})\nu^{i}\eta^{j}
+\frac{1}{4\pi}\int_{S_{\infty}}(E_{i}\nu^{i})(\vec{A}_{j}\eta^{j}),
\end{equation}
\begin{equation}\label{8}
q_{e} = \frac1{4\pi} \int_{S_\infty} E_i \nu^i\, , \qquad
q_{b} = \frac1{4\pi} \int_{S_\infty} B_i \nu^i\, ,
\end{equation}
where $S_{\infty}$ indicates the limit as $r\rightarrow\infty$ of integrals over coordinate spheres $S_{r}$, with unit
outer normal $\nu$, and $\vec{A}$ is the vector potential for the magnetic field. Due to topological considerations some care must be taken to construct the vector potential, moreover its contribution to \eqref{7} vanishes under appropriate asymptotic conditions \cite{DainKhuriWeinsteinYamada}; thus the current definition of angular momentum typically agrees the with the standard ADM notion. Here $q_{e}$ and $q_{b}$ denote the total electric and magnetic charge respectively,
and we denote the square of the total charge by $q^{2}=q_{e}^{2}+q_{b}^{2}$. Note that the fall-off in \eqref{3} is not
strong enough to imply that the ADM linear momentum vanishes, as is
typically assumed in the study of mass-angular momentum type inequalities. Therefore the expression \eqref{6}, which
represents the ADM energy, does not necessarily coincide with the standard definition of ADM mass as the length of the
4-momentum. Nevertheless, here, we will continue to refer to \eqref{6} as the mass. Furthermore, note that the asymptotics \eqref{3} are not necessarily strong enough to guarantee that the angular momentum is finite, since the Killing fields grow like $r$. However, under the addition hypothesis that $J_{\EM}(\eta)\in L^{1}(M_{\text{end}})$ it follows that \eqref{7} is finite, as may be seen from the proof of Lemma 2.1 in \cite{DainKhuriWeinsteinYamada}.

In the presence of an electromagnetic field, angular momentum is conserved \cite{Dain}, \cite{DainKhuriWeinsteinYamada}
if
\begin{equation}\label{9}
J_\EM^{i}\eta_{i}=0.
\end{equation}
In this case
\begin{equation}\label{14}
\mathcal{J}=\sum_{n=1}^{N}\mathcal{J}_{n},
\end{equation}
where $\mathcal{J}_{n}$ represents the angular momentum of the black hole $p_{n}$. Moreover, it will
be shown in the next section that the condition \eqref{9} gives rise to a charged twist potential $v$
which encodes the angular momentum by
\begin{equation}\label{13}
\mathcal{J}_{n}=\frac{1}{4}(v|_{I_{n}}-v|_{I_{n-1}}),
\end{equation}
where $I_{n}$ denotes the interval of the $z$-axis between $p_{n}$ and $p_{n+1}$, where $p_{0}=-\infty$ and
$p_{N+1}=\infty$. Potentials $\chi$ and $\psi$ may also be obtained for the electric and magnetic fields, respectively,
as a result of the constraints $\operatorname{div}_{g}E=\operatorname{div}_{g}B=0$. Similarly, the charges of
each black hole are given by
\begin{equation}\label{13.1}
q^{e}_{n}=\frac{1}{2}(\chi|_{I_{n}}-\chi|_{I_{n-1}}),\text{ }\text{ }\text{ }\text{ }\text{ }
q^{b}_{n}=\frac{1}{2}(\psi|_{I_{n}}-\psi|_{I_{n-1}}),
\end{equation}
with total charges
\begin{equation}\label{14.1}
q^{e}=\sum_{n=1}^{N}q^{e}_{n},\text{ }\text{ }\text{ }\text{ }\text{ }
q^{b}=\sum_{n=1}^{N}q^{b}_{n}.
\end{equation}

In the case of a single black hole, the mass-angular momentum-charge inequality \eqref{0} may be established in two
steps \cite{ChruscielCosta}, \cite{Costa}, \cite{SchoenZhou}. The first consists of proving a lower bound for the ADM
mass in terms of a harmonic map energy functional
\begin{equation}\label{20}
m\geq\mathcal{M}(U,v,\chi,\psi),
\end{equation}
where
\begin{equation}\label{21}
\mathcal{M}(U,v,\chi,\psi)
=\frac{1}{8\pi}\int_{\mathbb{R}^{3}}\left(|\nabla U|^{2}
+\frac{e^{4U}}{\rho^{4}}|\nabla v
+\chi\nabla\psi-\psi\nabla\chi|^{2}
+\frac{e^{2U}}{\rho^{2}}\left(|\nabla\chi|^{2}
+|\nabla\psi|^{2}\right)\right)dx
\end{equation}
with $|\nabla U|^2=(\partial_{\rho}U)^2+(\partial_{z}U)^2$ and $dx$ denoting the Euclidean volume element; this notation will be used throughout the paper.
The inequality \eqref{20} relies heavily on the assumption of a maximal data set $\tr_{g}k=0$, however
proposals for treating the nonmaximal case have been recently put forward in \cite{ChaKhuri}, \cite{ChaKhuri1}.
The second step entails showing that the data arising from the extreme Kerr-Newman spacetime
$(U_\KN,v_\KN,\chi_\KN,\psi_\KN)$ (see Appendix \ref{sec6}), minimize the functional among
all data with common angular momentum and charge
\begin{equation}\label{22}
\mathcal{M}(U,v,\chi,\psi)\geq\mathcal{M}(U_\KN,v_\KN,\chi_\KN,\psi_\KN).
\end{equation}
Since the right-hand side of \eqref{22} agrees with the square root of the right-hand side of \eqref{0},
together with \eqref{20} the desired conclusion is reached.

It should be pointed out that the hypotheses used in \cite{ChruscielCosta}, \cite{Costa}, and \cite{SchoenZhou} are
unnecessarily strong. In these works it is assumed that the matter density is nonnegative $\mu_\EM\geq 0$, the current
density vanishes $|J_\EM|_{g}=0$, and that the 4-currents for the electric and magnetic fields (sources for the Maxwell
equations) vanish. The later assumption concerning the 4-currents is imposed in order to secure the existence of
potentials for the Maxwell field, and $|J_\EM|_{g}=0$ is used to obtain a charged twist potential. Note that the use of
4-currents in general requires reference to an axisymmetric spacetime, as opposed to the initial data alone. This is
justified, since in electrovacuum the existence of an axisymmetric evolution of the initial data follows from its
smoothness \cite{Choquet-Bruhat}, \cite{Chrusciel0}. For our purposes, however, reference to the spacetime can be
avoided since we will show that the potentials arise in a direct manner from the initial data, under the weakened
hypotheses $\operatorname{div}_{g}E=\operatorname{div}_{g}B=J_\EM(\eta)=0$.

\begin{theorem}\label{thm1}
Let $(M,g,k,E,B)$ be a smooth, simply connected, axially symmetric, maximal initial data set satisfying $\mu_\EM\geq 0$
and $J_\EM(\eta)=0$, and with two ends, one designated asymptotically flat and the other either asymptotically flat or
asymptotically cylindrical.
Then
\begin{equation}\label{23}
m^2\geq\frac{q^{2}+\sqrt{q^{4}+4\mathcal{J}^{2}}}{2},
\end{equation}
and equality holds if and only if $(M,g,k,E,B)$ is isometric to the canonical slice of an extreme Kerr-Newman spacetime.
\end{theorem}

We point out that the rigidity statement of this result does not seem to have been properly established in the
literature, even in the uncharged case. What has been previously established \cite{SchoenZhou}, is that in the case of equality the map
into complex hyperbolic space arising from the given data agrees with the extreme Kerr-Newman harmonic map.

In the case of multiple black holes, the first step leading to \eqref{20} may be established using the same
arguments as those in the single black hole case. Thus, it is in the second step \eqref{22} where the significant
difference occurs. Here the minimizing harmonic map no longer arises from the extreme Kerr-Newman solution, or any
other well known black hole solution in general. An exception happens in the special situation when all charges have the
same sign and the angular momenta vanish, in which case the minimizing harmonic map arises from the Majumdar-Papapetrou
solution. In the generic case, a solution $(U_{0},v_{0},\chi_{0},\psi_{0})$ to the harmonic map equations is constructed
which has similar asymptotic behavior to that of the extreme Kerr-Newman map near each puncture $p_{n}$ and at the
designated asymptotically flat end. This asymptotic behavior allows an application of the convexity arguments
in \cite{SchoenZhou}, showing that the constructed solution minimizes the functional $\mathcal{M}$ and
yields a gap bound (see Theorem \ref{minimum}). Let
\begin{equation}\label{24}
\mathcal{F}(\mathcal{J}_{1},\ldots,\mathcal{J}_{N},q^{e}_{1},\ldots,q^{e}_{N},
q^{b}_{1},\ldots,q^{b}_{N})=\mathcal{M}(U_{0},v_{0},\chi_{0},\psi_{0})
\end{equation}
denote the minimum value of the functional.
Our main result is as follows.

\begin{theorem}\label{thm2}
Let $(M,g,k,E,B)$ be a smooth, simply connected, axially symmetric, maximal initial data set satisfying $\mu_\EM\geq 0$
and $J_\EM(\eta)=0$, and with $N+1$ ends, one designated asymptotically flat and the others either asymptotically flat
or asymptotically cylindrical.
Then
\begin{equation}\label{25}
m\geq\mathcal{F}(\mathcal{J}_{1},\ldots,\mathcal{J}_{N},q^{e}_{1},\ldots,q^{e}_{N},
q^{b}_{1},\ldots,q^{b}_{N}).
\end{equation}
\end{theorem}

The functions $\mathcal{F}$ appearing in \eqref{0.1} and \eqref{25} agree when the charges vanish.
Hence, Theorem \ref{thm2} generalizes the result of \cite{ChruscielLiWeinstein} by including charge, and slightly
improves this previous result in that the asymptotic assumptions on $k$ have been weakened.
Whether or not the right-hand side of \eqref{25} agrees with the square root of the right-hand side of \eqref{0}
is an important open question. Note that the case of equality is not addressed in Theorem \ref{thm2}, and is closely
related to the existence question for multiple rotating black hole solutions to the axisymmetric stationary
electrovacuum Einstein equations. In fact we will present arguments, based on the mass gap bound of Theorem \ref{minimum}, which suggest that generically equality cannot be
achieved in \eqref{25} when $N>1$.

\begin{conjecture}\label{con1}
Under the hypotheses of Theorem \ref{thm2}, equality in \eqref{25} cannot be achieved if $N>1$ unless
all charges are of the same sign and the angular momenta vanish. In this special case, the initial data set is
isometric to the canonical slice of a Majumdar-Papapetrou spacetime.
\end{conjecture}

This paper is organized as follows. In the next section we describe a deformation of the Maxwell field suited
for the existence of potentials, and prove Theorem \ref{thm1}. Section \ref{sec3} will be devoted to the construction
of a minimizer for the harmonic map functional in the case of multiple black holes, and appropriate estimates will be
established. In Section \ref{sec4}, Theorem \ref{thm2} will be proven and arguments supporting Conjecture \ref{con1}
will be given. The heuristic arguments leading to \eqref{0} will be discussed and extended in Appendix \ref{sec5}, to
the case when several black holes are moving apart at high velocities. Lastly Appendix \ref{sec6} is included to record
several formulae associated with the Kerr-Newman and Majumdar-Papapetrou spacetimes.

\section{Proof of Theorem \ref{thm1}}\label{sec2}

We first describe the construction of potentials, as alluded to in the introduction. Let
\begin{equation}\label{26}
e_{1}=e^{U-\alpha}(\partial_{\rho}-A_{\rho}\partial_{\phi}),
\text{ }\text{ }\text{ }\text{ } e_{2}=e^{U-\alpha}(\partial_{z}-A_{z}\partial_{\phi}),
\text{ }\text{ }\text{ }\text{ } e_{3}=\frac{1}{\sqrt{g_{\phi\phi}}}\partial_{\phi},
\end{equation}
be an orthonormal frame for $(M,g)$, with dual coframe
\begin{equation}\label{27}
\theta^{1}=e^{-U+\alpha}d\rho,
\text{ }\text{ }\text{ }\text{ }
\theta^{2}=e^{-U+\alpha}dz,
\text{ }\text{ }\text{ }\text{ }
\theta^{3}=\rho e^{-U}(d\phi+A_{\rho}d\rho+A_{z}dz).
\end{equation}
As in \cite{ChaKhuri1}, consider the projections $(\overline{E},\overline{B})$ of the electric and magnetic fields to the orbit space of
$\eta$, and let $F$ be the associated field strength defined on the auxiliary spacetime $(\mathbb{R}\times
M,-dt^{2}+g)$. That is
\begin{equation}\label{28}
\overline{E}(e_{i})=E(e_{i}), \text{ }\text{ }\text{ }\text{ } \overline{B}(e_{i})=B(e_{i}),\text{ }\text{ }\text{
}\text{ }i=1,2,\text{ }\text{ }\text{ }\text{ }\overline{E}(e_{3})=\overline{B}(e_{3})=0,
\end{equation}
and
\begin{equation} \label{29}
F(e_{i},\partial_{t})= \overline{E}_{i}, \text{ }\text{ }\text{ }\text{ }\text{ }
F(e_{i},e_{j})= \epsilon_{ijl}\overline{B}^{l}, \text{ }\text{ }\text{ }\text{ }\text{ } i,j,l = 1,2,3.
\end{equation}
Then
\begin{equation} \label{30}
F(\eta, \cdot)
= |\eta | \left(B(e_{2})\theta^{1}
-  B(e_{1})\theta^{2}\right),
\text{ }\text{ }\text{ }\text{ }\text{ }
\ast F(\eta, \cdot)
=|\eta | \left(E(e_{2})\theta^{1}
- E(e_{1})\theta^{2}\right),
\end{equation}
where $\ast$ denotes the Hodge star operation. It follows that
\begin{equation}\label{31}
d (F(\eta, \cdot) )=|\eta|(\operatorname{div}_{g}B)
\theta^{2}\wedge\theta^{1}=0,\text{ }\text{ }\text{ }\text{ }\text{ }
d(\ast F(\eta, \cdot))=|\eta|(\operatorname{div}_{g}E)
\theta^{2}\wedge\theta^{1}=0,
\end{equation}
and hence there exist potentials for the electromagnetic field such that
\begin{equation}\label{32}
d\psi=F(\eta, \cdot),\text{ }\text{ }\text{ }\text{ }\text{ }d\chi
=\ast F (\eta, \cdot).
\end{equation}
Moreover, a calculation (Lemma 4.1 of \cite{ChaKhuri1}) shows that
\begin{equation}\label{33}
d\left(k(\eta)\times\eta-\chi d\psi
+\psi d\chi\right)=|\eta|\left(J_\EM(\eta)-\chi \operatorname{div}_{g}B
+\psi \operatorname{div}_{g}E\right)\theta^{2}\wedge\theta^{1}=0,
\end{equation}
yielding a charged twist potential satisfying
\begin{equation}\label{34}
dv=k(\eta)\times\eta-\chi d\psi
+\psi d\chi.
\end{equation}
As mentioned in the introduction, the advantage of these computations is that they are made directly from
the initial data, and do not require reference to the evolved spacetime. They also show clearly that
the conditions $\operatorname{div}_{g}E=\operatorname{div}_{g}B=J_\EM(\eta)=0$ are necessary and sufficient for the
existence of the
desired potentials when $M$ is simply connected.

It should be noted that \eqref{30} and \eqref{32} imply that $\chi$ and $\psi$ are constant on each interval $I_{n}$ of
the $z$-axis, and
\begin{align}\label{34.1}
\begin{split}
q^{e}=\frac{1}{4\pi} \int_{S_\infty} E_i \nu^i
&=\sum_{n=1}^{N}\lim_{r_{n}\rightarrow 0}\frac{1}{4\pi} \int_{\partial B_{r_{n}}(p_{n})} E_i \nu^i dA\\
&=\sum_{n=1}^{N}\lim_{r_{n}\rightarrow 0}\frac{1}{4\pi} \int_{\partial B_{1}(p_{n})} E_i \nu^i
e^{2U-\alpha}r_{n}^{2}\sin\theta d\theta d\phi\\
&=-\sum_{n=1}^{N}\lim_{r_{n}\rightarrow 0}\frac{1}{4\pi} \int_{\partial B_{1}(p_{n})} |\eta|^{-1}
(\partial_{\theta}\chi) e^{-U}
r_{n}\sin\theta d\theta d\phi
=\sum_{n=1}^{N}\frac{1}{2}(\chi|_{I_{n}}-\chi|_{I_{n-1}})
\end{split}
\end{align}
where $B_{r}(p_{n})$ denotes the ball of radius $r$ centered at $p_{n}$. Similar computations yield the expressions for
$\mathcal{J}$ and $q^{b}$ in \eqref{13} and \eqref{13.1}.

From \eqref{30} and \eqref{32} we also find
\begin{equation}\label{35}
|\overline{E}|_{g}^{2}+|\overline{B}|_{g}^{2}
=\frac{e^{4U-2\alpha}}{\rho^{2}}\left(|\nabla\chi|^{2}
+|\nabla\psi|^{2}\right).
\end{equation}
Furthermore from \eqref{34} we have
\begin{equation}\label{36}
k(e_{1},e_{3})
=-|\eta|^{-2}\left(e_{2}(v)
+\chi e_{2}(\psi)
-\psi e_{2}(\chi)\right)
=-|\eta|^{-2}e^{U-\alpha}
\left(\partial_{z}v
+\chi\partial_{z}\psi
-\psi\partial_{z}\chi\right),
\end{equation}
and
\begin{equation}\label{37}
k(e_{2},e_{3})
=|\eta|^{-2}\left(e_{1}(v)
+\chi e_{1}(\psi)
-\psi e_{1}(\chi)\right)
=|\eta|^{-2}e^{U-\alpha}
\left(\partial_{\rho}v
+\chi\partial_{\rho}\psi
-\psi\partial_{\rho}\chi\right).
\end{equation}
It follows that
\begin{equation}\label{38}
|k|_{g}^{2}\geq 2\left(k(e_{1},e_{3})^{2}
+k(e_{2},e_{3})^{2}\right)
=2\frac{e^{6U-2\alpha}}{\rho^{4}}|\nabla v
+\chi\nabla\psi-\psi\nabla\chi|^{2}.
\end{equation}

Recall that in Brill coordinates, the scalar curvature may be expressed simply by (\cite{Brill}, \cite{Dain0})
\begin{equation}\label{39}
2e^{-2U+2\alpha}R=8\Delta U-4\Delta_{\rho,z}\alpha-4|\nabla U|^{2}-\rho^{2}e^{-2\alpha}
\left(A_{\rho,z}-A_{z,\rho}\right)^{2},
\end{equation}
where $\Delta$ is the Euclidean Laplacian on $\mathbb{R}^{3}$ and $\Delta_{\rho,z}=\partial_{\rho}^{2}
+\partial_{z}^{2}$. This leads to the following mass formula via an integration by parts
\begin{equation}\label{40}
m
=\frac{1}{32\pi}\int_{\mathbb{R}^{3}}\left(2e^{-2U+2\alpha}\!\text{ }R
+\rho^{2}e^{-2\alpha}
(A_{\rho,z}-A_{z,\rho})^{2}
+4|\nabla U|^{2}\right)dx.
\end{equation}
Observe that with the help of \eqref{35}, \eqref{38}, and the maximal data assumption, the scalar curvature may be
rewritten as
\begin{align}\label{41}
\begin{split}
R=&16\pi\mu_\EM+|k|_{g}^{2}+2\left(|E|_{g}^{2}+|B|_{g}^{2}\right)\\
=&16\pi\mu_\EM+\left(|k|_{g}^{2}-2k(e_{1},e_{3})^{2}-2k(e_{2},e_{3})^{2}\right)
+2\left(E(e_{3})^{2}+B(e_{3})^{2}\right)\\
&+2\frac{e^{6U-2\alpha}}{\rho^{4}}|\nabla v
+\chi\nabla\psi-\psi\nabla\chi|^{2}
+2\frac{e^{4U-2\alpha}}{\rho^{2}}\left(|\nabla\chi|^{2}
+|\nabla\psi|^{2}\right).
\end{split}
\end{align}
Therefore
\begin{align}\label{42}
\begin{split}
m-\mathcal{M}(U,v,\chi,\psi)
=&\frac{1}{32\pi}\int_{\mathbb{R}^{3}}\left(4e^{-2U+2\alpha}\left(E(e_{3})^2+B(e_{3})^2\right)
+\rho^{2}e^{-2\alpha}
(A_{\rho,z}-A_{z,\rho})^{2}\right)dx\\
&+\frac{1}{16\pi}\int_{\mathbb{R}^{3}}e^{-2U+2\alpha}\left(16\pi\mu_\EM
+\left(|k|_{g}^{2}-2k(e_{1},e_{3})^{2}-2k(e_{2},e_{3})^{2}\right)\right)dx.
\end{split}
\end{align}
It should be noted that this expression holds regardless of the number of ends. In the case of two ends
(\cite{ChruscielCosta}, \cite{Costa}, \cite{SchoenZhou}), \eqref{22} and \eqref{42} yield the inequality \eqref{23} in
Theorem \ref{thm1}.

Consider now the case of equality in \eqref{23}. From \eqref{22} and \eqref{42}, this implies that
\begin{equation}\label{43}
\mu_\EM=0,\text{ }\text{ }\text{ }\text{ }\text{ }\text{ }E(e_{3})=B(e_{3})=0,\text{ }\text{ }\text{ }\text{ }\text{
}\text{ }
A_{\rho,z}=A_{z,\rho},
\end{equation}
\begin{equation}\label{44}
\mathcal{M}(U,v,\chi,\psi)=\mathcal{M}(U_\KN,v_\KN,\chi_\KN,\psi_\KN),\text{ }\text{ }\text{ }\text{ }\text{
}k(e_{i},e_{j})=k(e_{3},e_{3})=0,\text{ }\text{ }\text{ }\text{ }\text{ }i,j\neq 3.
\end{equation}
According to the gap bound in \cite{SchoenZhou}, a map which minimizes the functional $\mathcal{M}$ must coincide with
the harmonic map associated with the extreme Kerr-Newman spacetime, that is
\begin{equation}\label{45}
(U,v,\chi,\psi)=(U_\KN,v_\KN,\chi_\KN,\psi_\KN).
\end{equation}
It follows immediately from \eqref{30}, \eqref{32}, and \eqref{43} that all components of the Maxwell field are
known and agree with those induced on the $t=0$ slice of the extreme Kerr-Newman spacetime, $(E,B)=(E_\KN,B_\KN)$.

Now observe that from \eqref{41}, \eqref{43}, \eqref{44}, and \eqref{45} we have $R=e^{-2\alpha}R_\KN$, where $R_\KN$ is
the scalar curvature of the $t=0$ slice of the extreme Kerr-Newman spacetime. Using the formula \eqref{39} produces
\begin{equation}\label{46}
e^{-2U_\KN}R_\KN=e^{-2U+2\alpha}R=4\Delta U_\KN-2|\nabla U_\KN|^{2}-2\Delta_{\rho,z}\alpha.
\end{equation}
However a direct computation from extreme Kerr-Newman data yields
\begin{equation}\label{47}
e^{-2U_\KN}R_\KN=e^{-2U+2\alpha}R=4\Delta U_\KN-2|\nabla U_\KN|^{2},
\end{equation}
so that $\Delta_{\rho,z}\alpha=0$. We claim that this, along with the asymptotics \eqref{17}, imply that $\alpha\equiv
0$. Note that it is sufficient to show that $\alpha=0$ along the z-axis. To see this, let $\vartheta\in(-\infty,2\pi)$ be the
cone angle deficiency \cite{Weinstein94} arising from the metric $g$ at the axis of rotation, that is
\begin{equation}\label{48}
\frac{2\pi}{2\pi-\vartheta}=\lim_{\rho\rightarrow 0}\frac{2\pi\cdot\mathrm{Radius}}{\mathrm{Circumference}}
=\lim_{\rho\rightarrow 0}\frac{\int_{0}^{\rho}e^{-U+\alpha}d\rho}{\rho e^{-U}}=e^{\alpha(0,z)}.
\end{equation}
Since $(M,g)$ is smooth across the axis of rotation, the angle deficiency must vanish $\vartheta=0$, and thus $\alpha(0,z)=0$. Note that the integral in \eqref{48} is not exactly $2\pi\cdot\mathrm{Radius}$, but is rather the top order approximation to this quantity; this is all that is needed to compute the desired limit.

We are now in a position to show that $(M,g)$ is isometric to the canonical slice of the extreme Kerr-Newman solution.
By \eqref{43} the 1-form $A_{\rho}d\rho+A_{z}dz$ is closed, and hence there exists a potential such that
$\partial_{\rho}f=A_{\rho}$ and $\partial_{z}f=A_{z}$. Consider the change of coordinates
$\widetilde{\phi}=\phi+f(\rho,z)$, then the metric \eqref{16} takes the form
\begin{equation}\label{49}
g=e^{-2U_\KN}(d\rho^{2}+dz^{2})+\rho^{2}e^{-2U_\KN}d\widetilde{\phi}^{2},
\end{equation}
which yields the desired result $g\cong g_\KN$. Lastly, observe that \eqref{36}, \eqref{37}, \eqref{44}, and
$\alpha(0,z)=0$ show that the tensor $k$ coincides with the extrinsic curvature of the canonical extreme Kerr-Newman
slice. This completes the proof of Theorem \ref{thm1}.

\section{Existence of the Minimizer and Estimates}\label{sec3}

In this section we prove the existence of a minimizer for the reduced energy \eqref{21},
having the asymptotics of extreme Kerr-Newman near each of the punctures $p_n$, and with prescribed angular momenta and
charges. The main tool will be Theorem~2
in~\cite{weinstein96}.
We denote by $\Gamma$ the $z$-axis in $\R^3$ and by $\Gamma'$ the axis $\Gamma$ minus the $N$ punctures $p_n$.
The model map $\tilde{\Phi}_{0}\colon\R^3\setminus\Gamma\to\H^2_\C$ which we construct below, is not
\emph{singular Dirichlet data} as defined in Definition~2 in~\cite{weinstein96}, because it does not satisfy
condition~(i). Nevertheless, this condition  is not used in the proof of the theorem, and the only key ingredient is
that the reduced energy of $\tilde{\Phi}_{0}$ must be finite and the pointwise tension of $\tilde{\Phi}_{0}$ must be
bounded with appropriate
decay at infinity, which will hold
true for the map constructed below. Consequently, Theorem~2 in~\cite{weinstein96} can be used to conclude that
there
is a harmonic map $\tilde{\Psi}_{0}\colon(\R^3\setminus\Gamma,g_{Euc})\to\H^2_\C$ which is \emph{asymptotic} to
$\tilde{\Phi}_{0}$, i.e.\ such that
$\dist_{\mathbb{H}^2_\C}(\tilde{\Phi}_{0},\tilde{\Psi}_{0})$ is bounded on $\R^3\setminus\Gamma$, and which satisfies the desired
boundary conditions on the axis.

The \emph{complex hyperbolic space} $\H^2_\C$ is the homogeneous Riemannian manifold $(\R^4,ds^2)$ where the metric is
given by
\begin{equation} \label{complex-hyperbolic}
  ds^2 = du^2 + e^{4u} (dv + \chi d\psi - \psi d\chi)^2 + e^{2u} (d\chi^2 + d\psi^2).
\end{equation}
Thus the energy density of a map $\tilde{\Phi}=(u,v,\chi,\psi)$ into $\H^2_\C$ is
\begin{equation}
  \mathcal{E}(\tilde{\Phi}) = |\nabla u|^2 + e^{4u} |\nabla v + \chi \nabla\psi - \psi \nabla\chi|^2 + e^{2u} (|\nabla
\chi|^2
+ |\nabla \psi|^2).
\end{equation}
Note that the metric~\eqref{complex-hyperbolic} is invariant under the \emph{translations}
\begin{equation}
  v \mapsto v - c\psi+ b\chi + a, \quad \psi\mapsto\psi+b,\quad \chi\mapsto\chi+c,
\end{equation}
for any constants $a,b,c$. Furthermore, the action of these translations is \emph{transitive} on any slice
${\mathbb S}_u=\{u=\text{constant}\}$, that is, ${\mathbb S}_u$ consists of a single orbit of this group action. The
\emph{tension} of the map $\tilde{\Phi}$ is the vector field $(\tau^u, \tau^v, \tau^\chi,\tau^\psi)$ on the pull-back
bundle
given by
\begin{align}
\begin{split}
  \tau^u &= \Delta u - 2 e^{4u} |\nabla v + \chi \nabla \psi - \psi \nabla\chi|^2 - e^{2u} \left(
  |\nabla \chi|^2  +  |\nabla \psi|^2 \right), \\
 \tau^v &= \Delta v + 2 \nabla u \cdot\nabla v + 2 (\nabla v + \chi \nabla\psi - \psi \nabla\chi) \cdot\nabla u - 2
e^{2u}
 (\nabla v + \chi \nabla\psi - \psi \nabla\chi)\cdot (\chi \nabla\chi + \psi \nabla\psi),\\
 \tau^\chi &= \Delta \chi + 2 \nabla u \cdot\nabla \chi - 2 e^{2u} (\nabla v + \chi \nabla\psi - \psi \nabla\chi)
 \cdot\nabla\psi,\\
 \tau^\psi &= \Delta\psi + 2\nabla u \cdot\nabla \psi + 2 e^{2u} (\nabla v + \chi \nabla\psi - \psi \nabla\chi)
\cdot\nabla
 \chi.
\end{split}
\end{align}
Thus the tension of a map vanishes if and only if the map is harmonic.

Note that for each set of values $({\mathcal J}_n,q^e_n,q^b_n)$, $n=1,\ldots,N$, there is an extreme Kerr-Newman
solution with this angular momentum, electric and magnetic charge, and there is a corresponding harmonic map
$\tilde{\Phi}_n\colon\R^3\setminus\Gamma\to\H^2_\C$. In view of the transitive isometric action of the translations
above, this
map $\tilde{\Phi}_n=(u_n,v_n,\chi_n,\psi_n)$ is only determined up to constants $(a_n,b_n,c_n)$, where $v_n\mapsto
v_n-c_n\psi_n+b_n\chi_n+a_n$, $\chi_n\mapsto\chi_n+b_n$, and $\psi_n\mapsto\psi_n+c_n$, as well as a domain
translation $z\mapsto z+d$. Thus, we can set the constant values of $(v_n,\chi_n,\psi_n)$ on the component of $\Gamma'$
on
one side of the puncture $p_n$, and the values on the other side of $p_n$ will be determined by the angular momentum,
electric and magnetic
charge of $\tilde{\Phi}_n$. Using this freedom we obtain $N+1$ harmonic maps $\tilde{\Phi}_n$,
where for each $n=1,\dots N$
the values of $(v_n,\chi_n,\psi_n)$ agree with the values of $(v,\chi,\psi)$ in the map $\tilde{\Psi}$ corresponding to
our given initial data set
on the components of $\Gamma'$ lying to both sides of the puncture $p_n$, and where the
values of $(v_{N+1},\chi_{N+1},\psi_{N+1})$ agree with those of $(v,\chi,\psi)$ on the unbounded components
$\Gamma_{\pm}'$ of the
axis $\Gamma'$. Furthermore, we still have an overall
translation available, i.e.\ a translation in $\H_\C^2$ which can be applied to $\tilde{\Psi}$. This will be used in the
proof of Lemma~\ref{model} below.

Let us now construct the model map $\tilde{\Phi}_{0}$. For the sake of simplicity, we illustrate the construction when
$N=2$ with the help of Figure~\ref{domain}; clearly the
construction can be carried out with any value of $N$. First, define the map $\tilde{\Phi}_{0}$ to be equal to
$\tilde{\Phi}_3$ on
${\mathcal B}$, the region outside a large ball which does not intersect a neighborhood of the punctures, and to be
equal to $\tilde{\Phi}_n$ on small balls ${\mathcal B}_n$ surrounding the punctures, $n=1,2$. These are
the dark shaded regions in
Figure~\ref{domain}. Next, define the map $\tilde{\Phi}_{0}$ in narrow cylindrical tubes surrounding the components of
$\Gamma'$
joining the
different ${\mathcal B}_n$'s, the lightly shaded regions in Figure~\ref{domain}. Consider for example ${\mathcal C}_2$.
We pick a smooth function $0\leq\lambda(z)\leq1$ which is $0$ near ${\mathcal B}_1$ and $1$ near ${\mathcal B}_2$,
and define
\begin{equation}
 \tilde{\Phi}_{0}=(1-\lambda)\tilde{\Phi}_1 + \lambda\tilde{\Phi}_2 = \bigl(
 (1-\lambda)u_1 + \lambda u_2, (1-\lambda)v_1 + \lambda v_2,(1-\lambda)\chi_1 + \lambda\chi_2, (1-\lambda)\psi_1 +
 \lambda\psi_2\bigr).
\end{equation}
Finally, we extend $\tilde{\Phi}_{0}$ to the remaining region $\Omega$, the white region in Figure~\ref{domain}, so that
it is smooth.

\begin{figure}
\includegraphics[width=10cm]{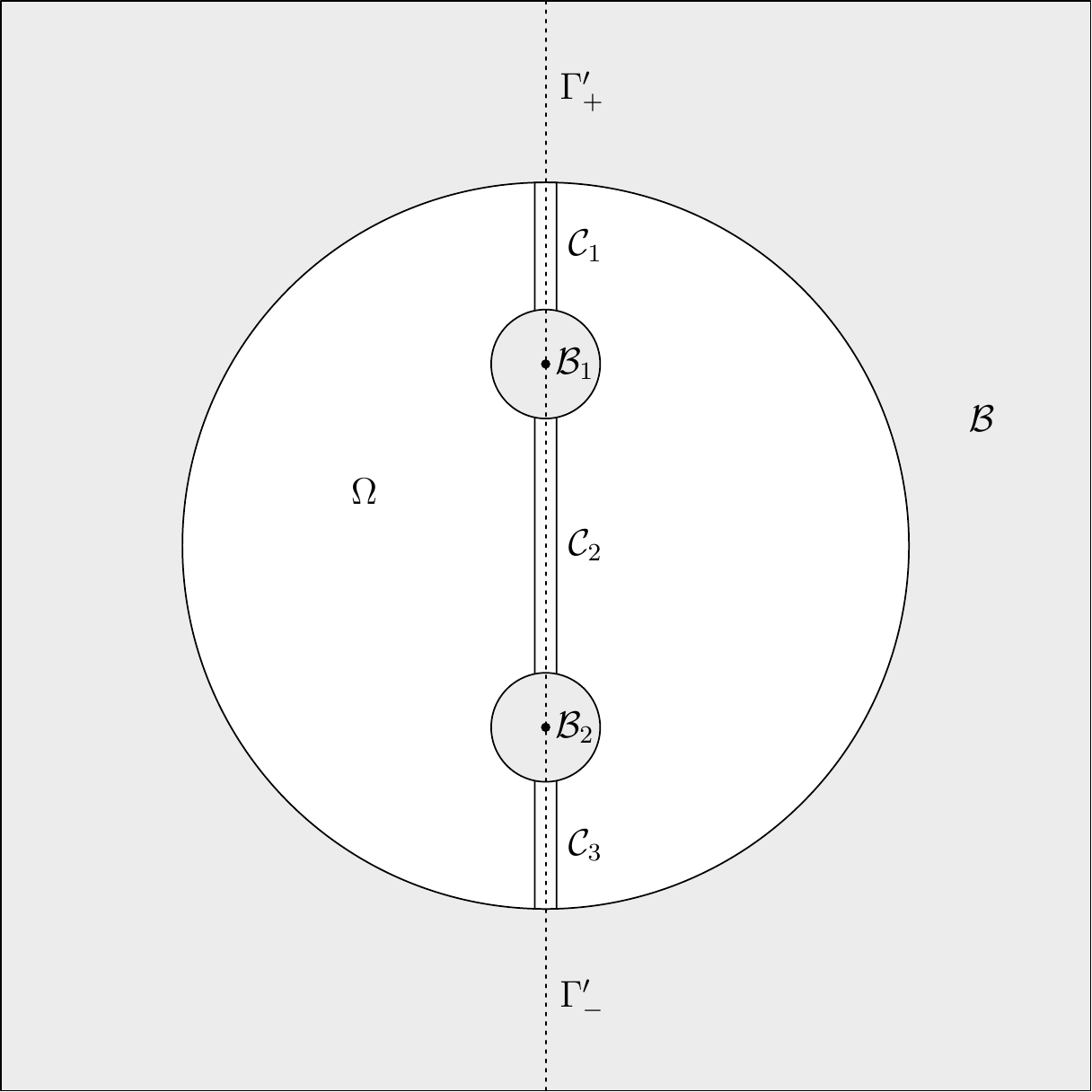}
 \caption{Construction of $\tilde\Phi_0$}  \label{domain}
\end{figure}

\begin{lemma} \label{model}
The reduced energy of $\tilde{\Phi}_{0}=(\tilde{u}_{0},\tilde{v}_{0},\tilde{\chi}_{0},\tilde{\psi}_{0})$ is finite, the
tension $\tau(\tilde{\Phi}_{0})$ has support inside a bounded set, and
$\tau(\tilde{\Phi}_{0})$ is pointwise bounded. Moreover, the values of
$(\tilde{v}_{0},\tilde{\chi}_{0},\tilde{\psi}_{0})$ agree with those of the given data $(v,\chi,\psi)$ on each component
of $\Gamma'$.
\end{lemma}

\begin{proof}
The reduced energy of the extreme Kerr-Newman harmonic map is finite, see for example~\cite{ChruscielCosta, Costa,
SchoenZhou}. Thus
the integral of the reduced energy density over the regions ${\mathcal B}_n$ is clearly finite. Also
the integral over $\Omega$ is finite. Thus, it is only necessary to check the integral over
${\mathcal C}_n$. For clarity we set $n=1$. Since all the quantities we seek to estimate are geometric invariants, we
can now use the last translation available to set all the constants $v_1|_{I_1}=v_2|_{I_{1}}$,
$\chi_1|_{I_{1}}=\chi_2|_{I_{1}}$,
$\psi_1|_{I_{1}}=\psi_2|_{I_{1}}$ to zero, where $I_{1}$ is the portion of the axis
between $p_1$ and $p_2$. This implies that $\rho^{-2}|v_n|$, $\rho^{-1}|\psi_n|$ and $\rho^{-1}|\chi_n|$, $n=1,2$, are
bounded in ${\mathcal C}_1$. We write $\tilde{U}_{0}=\tilde{u}_{0}+\log\rho$ and $U_n=u_n+\log\rho$, $n=1,2$. It follows
that
$\tilde{U}_{0}=(1-\lambda)U_1+\lambda U_2$. Thus, in
${\mathcal C}_1$, we have for the
reduced energy density
\begin{align}
\begin{split}
  \mathcal{E}'(\tilde{\Phi}_{0}) =& |\nabla \tilde{U}_{0}|^2
  +\frac{e^{4\tilde{U}_{0}}}{\rho^{4}}|\nabla \tilde{v}_{0}
 +\tilde{\chi}_{0}\nabla\tilde{\psi}_{0}-\tilde{\psi}_{0}\nabla\tilde{\chi}_{0}|^{2}
 +\frac{e^{2\tilde{U}_{0}}}{\rho^{2}}\left(|\nabla\tilde{\chi}_{0}|^{2}
 +|\nabla\tilde{\psi}_{0}|^{2}\right) \\
 \leq &|\nabla U_1|^2 + |\nabla U_{2}|^2 + C\left(\rho^{-4}|\nabla v_1|^2
 + \rho^{-2}|\nabla\chi_1|^2 + \rho^{-2}|\nabla\psi_1|^{2} + \dots
  + |\nabla \lambda|^2 \right),
 \end{split}
 \end{align}
where $C$ is a constant which depends on the supremum of $|U_n|$, $\rho^{-2}|v_n|$, $\rho^{-1}|\chi_n|$,
$\rho^{-1}|\psi_{n}|$.
Thus $\mathcal{E}'(\tilde{\Phi}_{0})$ is bounded over ${\mathcal C}_1$, and integrating over ${\mathcal C}_1$ clearly
gives a
finite quantity. It follows that the integral of the reduced energy over of ${\mathcal C}_n$, $n=2,3$
is also finite, and hence the reduced energy of $\tilde{\Phi}_{0}$ is finite.

Next, consider the second claim of the lemma.
Since the $\tilde{\Phi}_n$ are harmonic, $\tau(\tilde{\Phi}_n)=0$ on ${\mathcal B}_{n}$, $n=1,2$, and
$\tau(\tilde{\Phi}_3)=0$ on $\mathcal{B}$. Therefore the support of $\tau(\tilde{\Phi}_{0})$ is contained in
$\Omega\cup_{n=1}^{3}\mathcal{C}_{n}$.

For the proof of the third claim of the lemma, note that since the tension vanishes on ${\mathcal B}_n$ and is clearly
bounded on $\Omega$, it remains to check the boundedness on ${\mathcal
C}_n$. Once again, we focus on ${\mathcal C}_1$. Since $\tilde{\Phi}_n$ is harmonic, we have $ |\Delta U_n| \leq
2\mathcal{E}'(\tilde{\Phi}_n)$, $n=1,2$. Moreover, on
${\mathcal C}_1$
\begin{equation}
\Delta \tilde{u}_{0}=(1-\lambda)\Delta U_1 + \lambda\Delta U_2 + 2\nabla \lambda \cdot\nabla(U_2-U_1) +
(U_2-U_1)\Delta\lambda,
\end{equation}
so that
\begin{align}
\begin{split}
  |\tau^{\tilde{u}_{0}}| \leq& |\Delta \tilde{u}_{0}| + 2\mathcal{E}'(\tilde{\Phi}_{0})\\
   \leq&
  |\Delta U_1| + |\Delta U_2| + 2|\nabla\lambda|
  \left( \mathcal{E}'(\tilde{\Phi}_1) + \mathcal{E}'(\tilde{\Phi}_2)\right) + C |\Delta\lambda| +
2\mathcal{E}'(\tilde{\Phi}_{0}) \\
  \leq& 2(1 + |\nabla\lambda|) \left( \mathcal{E}'(\tilde{\Phi}_1) + \mathcal{E}'(\tilde{\Phi}_2)\right) +
  C|\Delta\lambda| + 2\mathcal{E}'(\tilde{\Phi}_{0}),
\end{split}
\end{align}
where $C=\sup(|U_1|+|U_2|)$. It follows that $|\tau^{\tilde{u}_{0}}|$ is bounded.
Next observe that $e^{2\tilde{u}_{0}}|\nabla \tilde{v}_{0}|$, $e^{\tilde{u}_{0}}|\tilde{\chi}_{0}|$,
$e^{\tilde{u}_{0}}|\nabla\tilde{\chi}_{0}|$, $e^{\tilde{u}_{0}}|\tilde{\psi}_{0}|$ and
$e^{\tilde{u}_{0}}|\nabla\tilde{\psi}_{0}|$ are all
bounded on ${\mathcal C}_1$, since the same is true of $e^{2u_n}|\nabla v_n|$, $e^{u_n}|\chi_n|$,
$e^{u_n}|\nabla\chi_n|$, $e^{u_n}|\psi_n|$ and $e^{u_n}|\nabla\psi_n|$, $n=1,2$. Using this, and $e^{2u}|\Delta v_n|\leq
C \mathcal{E}'(\Phi_n)$, the proof for the $\tilde{v}_{0}$
component now proceeds in much the same way as for the $\tilde{u}_{0}$ component:
\begin{align}
\begin{split}
  e^{2\tilde{u}_{0}} |\tau^{\tilde{v}_{0}}| \leq& e^{2\tilde{u}_{0}} |\Delta \tilde{v}_{0}| +C
\mathcal{E}'(\tilde{\Phi}_{0})\\
   \leq&
  e^{2\tilde{u}_{0}} |\Delta v_1| + e^{2\tilde{u}_{0}} |\Delta v_2| + 2 |\nabla\lambda|
  \left( \mathcal{E}'(\tilde{\Phi}_1) + \mathcal{E}'(\tilde{\Phi}_2)\right) + C|\Delta\lambda| +
C\mathcal{E}'(\tilde{\Phi}_{0})\\
  \leq& (C + 2|\nabla\lambda|) \left( \mathcal{E}'(\tilde{\Phi}_1) + \mathcal{E}'(\tilde{\Phi}_2)\right) +
C|\Delta\lambda| +
  C\mathcal{E}'(\tilde{\Phi}_{0}).
\end{split}
\end{align}
It follows that $e^{2\tilde{u}_{0}}|\tau^{\tilde{v}_{0}}|$ is bounded on ${\mathcal C}_1$.
In a similar way, we obtain that $e^{\tilde{u}_{0}}|\tau^{\tilde{\psi}_{0}}|$ and
$e^{\tilde{u}_{0}}|\tau^{\tilde{\chi}_{0}}|$ are bounded on ${\mathcal C}_1$, and it follows
that
\begin{equation}
  \tau(\tilde{\Phi}_{0}) = (\tau^{\tilde{u}_{0}})^2 + e^{4\tilde{u}_{0}} (\tau^{\tilde{v}_{0}}
  + \tilde{\chi}_{0} \tau^{\tilde{\psi}_{0}} - \tilde{\psi}_{0} \tau^{\tilde{\chi}_{0}})^2
  + e^{2\tilde{u}_{0}} \left( (\tau^{\tilde{\chi}_{0}})^2 +
  (\tau^{\tilde{\psi}_{0}})^2 \right)
\end{equation}
is bounded on ${\mathcal C}_1$.

Lastly, it is immediately apparent from the construction that the values of the potentials for the model map agree with
those of the given data on the axis.
\end{proof}

\begin{cor}
For any set of punctures $p_{n}$ on the axis $\Gamma$ and prescribed constants $v_{0}|_{I_{n}}$,
$\chi_{0}|_{I_{n}}$, $\psi_{0}|_{I_{n}}$, $n=1,\ldots,N$, there exists a corresponding unique
harmonic map $\tilde{\Psi}_{0}=(u_{0},v_{0},\chi_{0},\psi_{0})$ which is asymptotic to $\tilde{\Phi}_{0}$, and
satisfies
\begin{equation}\label{3.14}
U_{0}=u_{0}+\log\rho=\log r_{n}+O(1),\text{ }\text{ }\text{ }\text{ }v_{0},\chi_{0},\psi_{0}=O(1),\text{ }\text{ }\text{
as }\text{ }\text{ }r_{n}\rightarrow 0.
\end{equation}
\end{cor}

\begin{proof}
The existence of $\tilde{\Psi}_{0}$ and the fact that it is asymptotic to $\tilde{\Phi}_{0}$,
follow from the main result in \cite{weinstein96} combined with Lemma \ref{model}. In order to
establish uniqueness, assume that there are two solutions $\tilde{\Psi}_{0}$ and $\tilde{\Psi}_{1}$. Since the target
space has negative curvature, the distance function
$f(x)=\operatorname{dist}_{\mathbb{H}_{\mathbb{C}}^{2}}(\tilde{\Psi}_{0}(x),\tilde{\Psi}_{1}(x))\geq 0$ is subharmonic
on $\mathbb{R}^{3}\setminus\Gamma$. From this, it follows by a maximum principle type argument (Proposition C.4 in
\cite{ChruscielLiWeinstein}) that $f\equiv 0$.
\end{proof}

As a consequence of the fact that $\tilde{\Psi}_{0}$ is asymptotic to the model map $\tilde{\Phi}_{0}$  at spatial
infinity, that is $\operatorname{dist}_{\mathbb{H}_{\mathbb{C}}^{2}}
(\tilde{\Phi}_{0}, \tilde{\Psi}_{0})\rightarrow 0$ as $r\rightarrow\infty$, we obtain
\begin{equation}\label{3.11}
|U_{0}|\leq C\text{ }\text{ }\text{ on }\text{ }\text{ }\mathbb{R}^{3}\setminus\cup_{n=1}^{N}B_{\delta}(p_{n}),
\text{ }\text{ }\text{ }\text{ }
|v_{0}|+|\chi_{0}|+|\psi_{0}|\leq C\text{ }\text{ }\text{ on }\text{ }\text{ }\mathbb{R}^{3}.
\end{equation}
It is expected \cite{ChruscielLiWeinstein} that more refined asymptotic fall-off estimates should hold for
$\tilde{\Psi}_{0}$, which are in line with \eqref{57}-\eqref{64}.

In order to show that the solution $\Psi_{0}$ minimizes the functional \eqref{21}, we will need
certain estimates concerning the asymptotics.

\begin{prop}\label{CLWestimates}
On $\mathbb{R}^{3}\setminus\Gamma$, the solution $\Psi_{0}=(U_{0},v_{0},\chi_{0},\psi_{0})$ satisfies the estimate
\begin{equation}\label{3.1}
|\nabla(U_{0}-\log\rho)|^{2}+\frac{e^{4U_{0}}}{\rho^{4}}
|\nabla v_{0}+\chi_{0}\nabla\psi_{0}-
\psi_{0}\nabla\chi_{0}|^{2}
+\frac{e^{2U_{0}}}{\rho^{2}}(|\nabla\chi_{0}|^{2}
+|\nabla\psi_{0}|^{2})\leq \frac{C}{\rho^{2}},
\end{equation}
and in particular near each puncture $p_{n}$ it holds that
\begin{equation}\label{3.2}
|\nabla U_{0}|\leq \frac{C}{\rho},\text{ }\text{ }\text{ }\text{ }
|\nabla v_{0}+\chi_{0}\nabla\psi_{0}-
\psi_{0}\nabla\chi_{0}|\leq \frac{C\rho}{r_{n}^{2}},\text{ }\text{ }\text{ }\text{ }
|\nabla\chi_{0}|
+|\nabla\psi_{0}|\leq\frac{C}{r_{n}}.
\end{equation}
\end{prop}

\begin{proof}
\newcommand{\E}{\mathcal{E}}
This result is analogous to (2.20) in \cite{ChruscielLiWeinstein}, and the proof follows in the same way. As details of the proof were left out in \cite{ChruscielLiWeinstein}, we provide them here in the current and more general setting.

Let $x$ be a point of Euclidean distance $\rho$ from the axis $\Gamma$, and let $B_{\rho/2}(x)$ be a Euclidean ball of radius $\rho/2$. By the Bochner identity (equation (12) in \cite{weinstein96}), and the fact that the target manifold $\mathbb{H}^{2}_{\mathbb{C}}$ has negative curvature, it follows that the energy density $\mathcal{E}(\tilde{\Psi}_0)$ is subharmonic $\Delta\mathcal{E}(\tilde{\Psi}_0)\geq0$. Thus, from the mean value theorem we obtain
\begin{equation}
\mathcal{E}(\tilde{\Psi}_0)|_{x} \leq \fint_{B_{\rho/2}(x)} \mathcal{E}(\tilde{\Psi}_0),
\end{equation}
where $\fint_{B_{\rho/2}} = \operatorname{vol}(B_{\rho/2})^{-1} \int_{B_{\rho/2}}$ denotes the average.
Now observe that $\mathcal{E}(\tilde{\Psi}_0) \leq \mathcal{E}'(\Psi_0)+C/\rho^2$, and taking the average of this inequality over $B_{\rho/2}(x)$ produces
\begin{equation}
	\mathcal{E}(\tilde{\Psi}_0)|_x \leq \fint_{B_{\rho/2}(x)} \E'(\Psi_0) + \frac C{\rho^2}.
\end{equation}
Similarly we have $\mathcal{E}'(\Psi_0) \leq  \mathcal{E}(\tilde{\Psi}_0) + C/\rho^2$, and hence
\begin{equation}
	\mathcal{E}'(\Psi_0)|_x \leq \fint_{B_{\rho/2}(x)} \E'(\Psi_0) + \frac C{\rho^2}.
\end{equation}
It remains only to show that $\fint_{B_{\rho/2}(x)} \E'(\Psi_0)$ is of order $1/\rho^2$.

We divide this argument into two cases: away from a neighborhood of the punctures, and in a neighborhood of the punctures. First assume that $x$ is away from any puncture. Then since $U_0$ is uniformly bounded, we can use the arguments of (10) and (11) in \cite[pages 842-843]{weinstein96} to obtain
\begin{equation}
	\int_{B_{\rho/2}(x)}\mathcal{E}'(\Psi_0) \leq \left(4\sup_{B_{3\rho/4}(x)}|U_0|^2 + 2\sup_{B_{3\rho/4}(x)}|U_0|\right) \int_{B_{3\rho/4}(x)} |\nabla\varphi|^2,
\end{equation}
where $\varphi$ is a cut-off function which is $1$ on $B_{\rho/2}(x)$ and vanishes outside $B_{3\rho/4}(x)$. Clearly $|\nabla\varphi|^2$ is of order $1/\rho^2$ in $B_{3\rho/4}(x)$, so after dividing by $\operatorname{vol}(B_{\rho/2}(x))$ the desired result follows.

Assume now that $x$ is in a neighborhood of one of the punctures, which without loss of generality may be taken to be the origin. From the estimate \eqref{3.14}, of $U_0$ near the puncture, we have $|U_0-\log r|\leq C$. Furthermore $C^{-1}r(x)\leq r\leq Cr(x)$ on $B_{3\rho/4}(x)$, so that $|\log r-\log r(x)|\leq \log C$ on this domain. It follows that $|U_0-\log r(x)|\leq C+\log C$. Since $\log r(x)$ is constant on $B_{3\rho/4}(x)$, we can replace $U_0$ by $U_0'=U_0-\log r(x)$ in the arguments above to obtain
\begin{equation}
	\int_{B_{\rho/2}(x)}\mathcal{E}'(\Psi_0) \leq \left(4\sup_{B_{3\rho/4}(x)}|U'_0|^2 + 2\sup_{B_{3\rho/4}(x)}|U'_0|\right) \int_{B_{3\rho/4}(x)} |\nabla\varphi|^2.
\end{equation}
This completes the proof of \eqref{3.1}. The estimates \eqref{3.2} follow immediately from \eqref{3.1} and \eqref{3.14}.
\end{proof}

The previous proposition is useful because the bounds it provides in \eqref{3.1} apply globally. However, more detailed estimates are
possible near any compact subset of the axis $\Gamma$ which does not include punctures. More precisely, $\Psi_{0}$ is
smooth away from the punctures and satisfies the following asymptotics \cite{Nguyen}:
\begin{equation}\label{3.5}
v_{0}=c_{1}+c_{3}\chi_{0}-c_{2}\psi_{0}+O(\rho^{4})
=c_{4}+O(\rho^{2}),
\text{ }\text{ }\text{ }\text{ }
\chi_{0}=c_{2}+O(\rho^{2}),\text{ }\text{ }\text{ }\text{ }\psi_{0}=c_{3}+O(\rho^{2}),
\end{equation}
\begin{equation}\label{3.6}
|\nabla(U_{0}-\log\rho)|=O(\rho^{-1+\epsilon}),\text{ }\text{ }\text{ }\text{ }\text{ }
|\Delta U_{0}|\leq C\rho^{-2+\epsilon},
\end{equation}
for some $\epsilon>0$, and
\begin{equation}\label{3.7}
|\nabla v_{0}|=O(\rho),\text{ }\text{ }\text{ }\text{ }
|\nabla\chi_{0}|+|\nabla\psi_{0}|=O(\rho),
\text{ }\text{ }\text{ }\text{ }
|\nabla v_{0}+\chi_{0}\nabla\psi_{0}-\psi_{0}\nabla\chi_{0}|=O(\rho^{3}).
\end{equation}

Lastly, we will have need of the following weighted Poincar\'{e} inequalities.

\begin{lemma}\label{poincare}
Let $\beta>\frac{1}{2}$ and $h\in C^{0}$. If $f\in C^{1}$ is axisymmetric and satisfies
$(\sin^{-\beta}\theta)f|_{\theta=0,\pi}=0$ then
\begin{equation}\label{3.10}
\int_{\delta_{1}<r<\delta_{2}}\frac{f^{2}}{\sin^{2\beta+1}\theta}\frac{h(r)}{r^{2}}
\leq
C\int_{\delta_{1}<r<\delta_{2}}\frac{|\nabla f|^{2}}{\sin^{2\beta-1}\theta}h(r)
\end{equation}
where we are using $r$ to denote the Euclidean distance to any puncture.

If $f\in C^{1}$ is an axisymmetric function with $f|_{r=\delta_{1},\delta_{2}}=0$, and $\beta\neq 0$ then
\begin{equation}\label{3.12}
\int_{\delta_{1}<r<\delta_{2}}\frac{f^{2}}{r^{2\beta+2}}
\leq C|\beta|^{-2}\int_{\delta_{1}<r<\delta_{2}}\frac{(\partial_{r}f)^{2}}{r^{2\beta}}.
\end{equation}

\end{lemma}

\begin{proof}
The first statement slightly generalizes Proposition 2.4 in \cite{ChruscielLiWeinstein}, but the proof there easily
extends to this situation by keeping track of boundary terms.

For the second statement, project everything to the $(\rho,z)$ plane so that the integration takes place over an annulus
$A(\delta_{1},\delta_{2})$. Now use the fact that $\log r$ is harmonic in two dimensions, and the fact that
$f|_{r=\delta_{1},\delta_{2}}=0$ to find
\begin{equation}\label{3.13}
0=\int_{A(\delta_{1},\delta_{2})}
\nabla(\log r)\cdot\nabla(r^{-2\beta}f^{2})
=\int_{A(\delta_{1},\delta_{2})}\frac{1}{r}
\left(-2\beta r^{-2\beta-1}f^{2}
+2r^{-2\beta}f\partial_{r}f\right).
\end{equation}
From this we easily get the desired inequality.
\end{proof}

\section{Convexity and Proof of Theorem \ref{thm2}}\label{sec4}

Let $\tilde{\Psi}=(u,v,\chi,\psi):\mathbb{R}^{3}\rightarrow
\mathbb{H}_{\mathbb{C}}^{2}$ and consider the harmonic energy on a domain $\Omega\subset\mathbb{R}^{3}$:
\begin{equation}\label{50}
E_{\Omega}(\tilde{\Psi})=\int_{\Omega}|\nabla u|^{2}+e^{4u}|\nabla v
+\chi\nabla\psi-\psi\nabla\chi|^{2}
+e^{2u}\left(|\nabla\chi|^{2}
+|\nabla\psi|^{2}\right)dx.
\end{equation}
If $\Omega$ does not intersect the rotation axis $\Gamma=\{\rho=0\}$, and we write $U=u+\log\rho$, then the reduced
energy $\mathcal{I}_{\Omega}$ of the map $\Psi=(U,v,\chi,\psi)$ is related to the harmonic energy of $\tilde{\Psi}$ by
\begin{equation}\label{51}
\mathcal{I}_{\Omega}(\Psi)=E_{\Omega}(\tilde{\Psi})
+\int_{\partial\Omega}(2U-\log\rho)\partial_{\nu}\log\rho,
\end{equation}
where $\nu$ denotes the unit outer normal to the boundary $\partial\Omega$ and
\begin{equation}\label{52}
\mathcal{I}_{\Omega}(\Psi)=
\int_{\Omega}|\nabla U|^{2}+\frac{e^{4U}}{\rho^{4}}|\nabla v
+\chi\nabla\psi-\psi\nabla\chi|^{2}
+\frac{e^{2U}}{\rho^{2}}\left(|\nabla\chi|^{2}
+|\nabla\psi|^{2}\right)dx.
\end{equation}
The formula \eqref{51} is obtained through an integration by parts, using the fact that $\log\rho$ is harmonic on
$\mathbb{R}^{3}\setminus\Gamma$. Note that $\mathcal{I}=\mathcal{I}_{\mathbb{R}^{3}}=8\pi\mathcal{M}$
where $\mathcal{M}$ was introduced in Section \ref{sec1}. Moreover $\mathcal{I}$, which is referred to as the reduced
energy, may be considered a regularization of $E$ since the infinite term $\int|\nabla\log\rho|^{2}$ has been removed,
and since the two functionals differ only by a boundary term they must have the same critical points.

Let $\tilde{\Psi}_{0}=(u_{0},v_{0},\chi_{0},\psi_{0})$ denote the harmonic map constructed in the previous section, and
let $\Psi_{0}=(U_{0},v_{0},\chi_{0},\psi_{0})$ be the associated renormalized map where $U_{0}=u_{0}+\log\rho$. Thus,
$\Psi_{0}$ is a critical point of $\mathcal{I}$. It is the purpose of this section to show that $\Psi_{0}$ realizes the
global minimum for $\mathcal{I}$.

\begin{theorem}\label{minimum}
Suppose that $\Psi=(U,v,\chi,\psi)$ is smooth and satisfies the asymptotics \eqref{17}-\eqref{19}, \eqref{57}-\eqref{59}
with $v|_{\Gamma}=v_{0}|_{\Gamma}$, $\chi|_{\Gamma}=\chi_{0}|_{\Gamma}$, $\psi|_{\Gamma}=\psi_{0}|_{\Gamma}$, then there
exists a constant $C>0$ such that
\begin{equation}\label{53}
\mathcal{I}(\Psi)-\mathcal{I}(\Psi_{0})
\geq C\left(\int_{\mathbb{R}^{3}}
\operatorname{dist}_{\mathbb{H}_{\mathbb{C}}^{2}}^{6}(\Psi,\Psi_{0})dx
\right)^{\frac{1}{3}}.
\end{equation}
\end{theorem}

This theorem is analogous to that of Theorem 6.1 in \cite{SchoenZhou}, where the role of extreme Kerr-Newman is now
played by the (possibly) multiple black hole solution $\Psi_{0}$ constructed in the previous section. The proof in
\cite{SchoenZhou} is based on convexity of the harmonic energy under geodesic deformations; such a property is true
under general circumstances when the target space is nonpositively curved. More precisely, let $\delta,\varepsilon>0$ be
small parameters and set $\Omega_{\delta,\varepsilon}=\{\delta<r_{n}\text{ for }n=1,\ldots,N; r<2/\delta;
\rho>\varepsilon\}$ and $\mathcal{A}_{\delta,\varepsilon}=B_{2/\delta}\setminus
\Omega_{\delta,\varepsilon}$, where $B_{2/\delta}$ is the ball of radius $2/\delta$ centered at the origin. Via a cut-and-paste argument, it will be shown that we may assume
$\Psi$ satisfies
\begin{equation}\label{54}
\operatorname{supp}(U-U_{0})\subset B_{2/\delta},\text{ }\text{ }\text{ }\text{ }\text{ }
\operatorname{supp}(v-v_{0},\chi-\chi_{0},\psi-\psi_{0})\subset \Omega_{\delta,\varepsilon}.
\end{equation}
If $\tilde{\Psi}_{t}$, $t\in[0,1]$, is a geodesic in $\mathbb{H}_{\mathbb{C}}^{2}$ connecting
$\tilde{\Psi}_{1}=\tilde{\Psi}$ and $\tilde{\Psi}_{0}$ (this means that for each $x$ in the domain, $t\rightarrow\tilde{\Psi}_{t}(x)$ is a geodesic), then $\Psi_{t}\equiv\Psi_{0}$ outside $B_{2/\delta}$ and
$(v_{t},\chi_{t},\psi_{t})\equiv(v_{0},\chi_{0},\psi_{0})$
in a neighborhood of $\mathcal{A}_{\delta,\varepsilon}$, so that in particular $U_{t}=U_{0}+t(U-U_{0})$ on these
domains. This simple expression for $U_{t}$ together with convexity of the harmonic energy yields
\begin{equation}\label{55}
\frac{d^{2}}{dt^{2}}\mathcal{I}(\Psi_{t})
\geq 2\int_{\mathbb{R}^{3}}|\nabla\operatorname{dist}_{\mathbb{H}^{2}_{\mathbb{C}}}(\Psi,\Psi_{0})|^{2}dx.
\end{equation}
Moreover, the fact that $\Psi_{0}$ is a critical point implies that
\begin{equation}\label{56}
\frac{d}{dt}\mathcal{I}(\Psi_{t})|_{t=0}=0.
\end{equation}
Theorem \ref{minimum} then follows by integrating \eqref{55} and applying a Sobolev inequality. In the remainder of this
section we will justify each of these steps, following closely the strategy of \cite{SchoenZhou} in the case of a single
black hole. Most of the effort required to establish each step consists of estimating certain integrals. Here, however,
the techniques used for these estimates will be significantly different since $\Psi_{0}$ is not known explicitly,
whereas in the single black hole case $\Psi_{0}$ is explicit as it arises from the extreme Kerr-Newman solution.

Before proceeding we record the appropriate asymptotic behavior of $\Psi$. Asymptotics for $U$ are given in \eqref{17},
\eqref{18}, \eqref{19}, and if $\omega=dv+\chi d\psi-\psi d\chi$ then
\begin{equation}\label{57}
|\omega|=\rho^{2}O(r^{-\lambda}),\text{ }\text{ }\text{ }\text{ }|\nabla\chi|+|\nabla\psi|=\rho O(r^{-\lambda})\text{
}\text{ }\text{ as }\text{ }\text{ }r\rightarrow\infty,
\end{equation}
\begin{equation}\label{58}
|\omega|=\rho^{2}O(r_{n}^{\lambda-6}),\text{ }\text{ }\text{ }|\nabla\chi|+|\nabla\psi|=\rho O(r_{n}^{\lambda-4})\text{
}\text{ }\text{ as }\text{ }\text{ }r_{n}\rightarrow 0\text{ }\text{ }\text{ in asymptotically flat ends},
\end{equation}
\begin{equation}\label{58.1}
|\omega|=\rho^{2}O(r_{n}^{\lambda-5}),\text{ }\text{ }|\nabla\chi|+|\nabla\psi|=\rho O(r_{n}^{\lambda-3})\text{ }\text{
as }\text{ }r_{n}\rightarrow 0\text{ }\text{ in asymptotically cylindrical ends},
\end{equation}
\begin{equation}\label{59}
|\omega|=O(\rho^{2}),\text{ }\text{ }\text{ }\text{ }|\nabla\chi|+|\nabla\psi|=O(\rho) \text{ }\text{ }\text{ as }\text{
}\text{ }\rho\rightarrow 0\text{ }\text{ }\text{ in }\text{ }\text{ }\Omega_{\delta,\varepsilon}.
\end{equation}
Note that with these asymptotics $\mathcal{I}(\Psi)$ is finite precisely when $\lambda>\frac{3}{2}$. Moreover, one may
integrate along lines perpendicular to $\Gamma$ to find
\begin{equation}\label{60}
|\chi|,|\psi|=\text{const}+\rho^{2}O(r^{-\lambda})\text{ }\text{ }\text{ as }\text{ }\text{ }r\rightarrow\infty,
\end{equation}
\begin{equation}\label{61}
|\chi|,|\psi|=\text{const}+\rho^{2}O(r_{n}^{\lambda-4})\text{ }\text{ }\text{ as }\text{ }\text{ }r_{n}\rightarrow
0\text{ }\text{ }\text{ in asymptotically flat ends},
\end{equation}
\begin{equation}\label{61.1}
|\chi|,|\psi|=\text{const}+\rho^{2}O(r_{n}^{\lambda-3})\text{ }\text{ }\text{ as }\text{ }\text{ }r_{n}\rightarrow
0\text{ }\text{ }\text{ in asymptotically cylindrical ends},
\end{equation}
\begin{equation}\label{61.2}
|\chi|,|\psi|=\text{const}+O(\rho^{2})\text{ }\text{ }\text{ as }\text{ }\text{ }\rho\rightarrow 0\text{ }\text{
}\text{ in }\text{ }\text{ }\Omega_{\delta,\varepsilon},
\end{equation}
from which it follows that
\begin{equation}\label{62}
|\nabla v|=\rho O(r^{-\lambda+1})\text{ }\text{ }\text{ as }\text{ }\text{ }r\rightarrow\infty,
\end{equation}
\begin{equation}\label{63}
|\nabla v|=\rho O(r_{n}^{\lambda-5})\text{ }\text{ }\text{ as }\text{ }\text{ }r_{n}\rightarrow 0\text{ }\text{ }\text{
in asymptotically flat ends},
\end{equation}
\begin{equation}\label{63.1}
|\nabla v|=\rho O(r_{n}^{\lambda-4})\text{ }\text{ }\text{ as }\text{ }\text{ }r_{n}\rightarrow 0\text{ }\text{ }\text{
in asymptotically cylindrical ends},
\end{equation}
\begin{equation}\label{64}
|\nabla v|=\rho O(r^{-\lambda+1})\text{ }\text{ }\text{ as }\text{ }\text{ }\rho\rightarrow 0\text{ }\text{ }\text{ in
}\text{ }\text{ }\Omega_{\delta,\varepsilon}.
\end{equation}

In order to carry out the proof of Theorem \ref{minimum} as outlined above, we must first show that it is possible to
approximate $\mathcal{I}(\Psi)$ by replacing $\Psi$ with a map that satisfies \eqref{54}. This may be achieved as in
\cite{SchoenZhou} with a three step cut and paste argument. Define smooth cut-off functions
\begin{equation}\label{65}
\varphi_{\delta}^{1}=\begin{cases}
1 & \text{ if $r\leq\frac{1}{\delta}$,} \\
|\nabla\varphi_{\delta}^{1}|\leq 2\delta &
\text{ if $\frac{1}{\delta}<r<\frac{2}{\delta}$,} \\
0 & \text{ if $r\geq\frac{2}{\delta}$,} \\
\end{cases}
\end{equation}
\begin{equation}\label{66}
\varphi_{\delta}=\begin{cases}
0 & \text{ if $r_{n}\leq\delta$,} \\
|\nabla\varphi_{\delta}|\leq \frac{2}{\delta} &
\text{ if $\delta<r_{n}<2\delta$,} \\
1 & \text{ if $r_{n}\geq2\delta$,} \\
\end{cases}
\text{ }\text{ }\text{ for }\text{ }\text{ }n=1,\ldots,N,
\end{equation}
\begin{equation}\label{67}
\phi_{\varepsilon}=\begin{cases}
0 & \text{ if $\rho\leq\varepsilon$,} \\
\frac{\log(\rho/\varepsilon)}{\log(\sqrt{\varepsilon}/
\varepsilon)} &
\text{ if $\varepsilon<\rho<\sqrt{\varepsilon}$,} \\
1 & \text{ if $\rho\geq\sqrt{\varepsilon}$.} \\
\end{cases}
\end{equation}
The first step deals with the region $M_{\text{end}}^{0}$. Let
\begin{equation}\label{68}
F_{\delta}^{1}(\Psi)=\Psi_{0}
+\varphi_{\delta}^{1}(\Psi-\Psi_{0})=:
(U_{\delta}^{1},v_{\delta}^{1},\chi_{\delta}^{1}
,\psi_{\delta}^{1}),
\end{equation}
so that $F_{\delta}^{1}(\Psi)=\Psi_{0}$ on $\mathbb{R}^{3}\setminus B_{2/\delta}$.

\begin{lemma}\label{converge1}
$\lim_{\delta\rightarrow 0}\mathcal{I}(F_{\delta}^{1}(\Psi))=\mathcal{I}(\Psi).$
\end{lemma}

\begin{proof}
Write
\begin{equation}\label{69}
\mathcal{I}(F_{\delta}^{1}(\Psi))
=\mathcal{I}_{r\leq\frac{1}{\delta}}(F_{\delta}^{1}(\Psi))
+\mathcal{I}_{\frac{1}{\delta}< r<\frac{2}{\delta}}(F_{\delta}^{1}(\Psi))
+\mathcal{I}_{r\geq\frac{2}{\delta}}(F_{\delta}^{1}(\Psi)),
\end{equation}
and observe that $\mathcal{I}_{r\leq\frac{1}{\delta}}(F_{\delta}^{1}(\Psi))
\rightarrow \mathcal{I}(\Psi)$ by the dominated convergence theorem (DCT) and  since $\Psi_{0}$ has finite reduced
energy $\mathcal{I}_{r\geq\frac{2}{\delta}}(F_{\delta}^{1}(\Psi))\rightarrow 0$. Now write
\begin{equation}\label{70}
\mathcal{I}_{\frac{1}{\delta}< r<\frac{2}{\delta}}(F_{\delta}^{1}(\Psi))
=\underbrace{\int_{\frac{1}{\delta}<r<\frac{2}{\delta}}
|\nabla U_{\delta}^{1}|^{2}}_{I_{1}}
+\underbrace{\int_{\frac{1}{\delta}<r<\frac{2}{\delta}}
\frac{e^{4U_{\delta}^{1}}}{\rho^{4}}|\omega_{\delta}^{1}|^{2}}_{I_{2}}
+\underbrace{\int_{\frac{1}{\delta}<r<\frac{2}{\delta}}
\frac{e^{2U_{\delta}^{1}}}{\rho^{2}}(|\nabla\chi_{\delta}^{1}|^{2}
+|\nabla\psi_{\delta}^{1}|^{2})}_{I_{3}}.
\end{equation}
We have
\begin{equation}\label{71}
I_{1}\leq 2\int_{\frac{1}{\delta}<r<\frac{2}{\delta}}
\left(|\nabla U|^{2}+|\nabla U_{0}|^{2}
+(U-U_{0})^{2}\underbrace{|\nabla\varphi_{\delta}^{1}|^{2}}
_{\leq 4\delta^{2}}\right),
\end{equation}
where the first two terms converge to zero by the DCT and finite reduced energy of $\Psi_{0}$, respectively. For the
third term we may apply H\"{o}lder's inequality and the Gagliardo-Nirenberg-Sobolev inequality to find
\begin{equation}\label{72}
\int_{\frac{1}{\delta}<r<\frac{2}{\delta}}\!\!
(U-U_{0})^{2}\underbrace{|\nabla\varphi_{\delta}^{1}|^{2}}
_{\leq 4\delta^{2}}
\leq \! \left(\int_{\frac{1}{\delta}<r<\frac{2}{\delta}}\!\!
(U-U_{0})^{6}\right)^{\frac{1}{3}}\!
\left(\int_{\frac{1}{\delta}<r<\frac{2}{\delta}}\!\!
|\nabla\varphi_{\delta}^{1}|^{3}\right)^{\frac{2}{3}}
\!\leq C\!\int_{\frac{1}{\delta}<r<\frac{2}{\delta}}\!\!
|\nabla (U-U_{0})|^{2}\rightarrow 0.
\end{equation}
Note that the Gagliardo-Nirenberg-Sobolev inequality applies here since $U,U_{0}\in H^{1}(\mathbb{R}^{3})$ (the Sobolev
space of square integrable derivatives) are limits of compactly supported functions.

Now consider $I_{2}$, and write
\begin{align}\label{73}
\begin{split}
\omega_{\delta}^{1}=&\varphi_{\delta}^{1}\omega
+(1-\varphi_{\delta}^{1})\omega_{0}
+(v-v_{0})\nabla\varphi_{\delta}^{1}
+(\chi_{0}\psi-\psi_{0}\chi)\nabla\varphi_{\delta}^{1}
\\
&+\varphi_{\delta}^{1}(1-\varphi_{\delta}^{1})
[(\psi-\psi_{0})\nabla(\chi-\chi_{0})
-(\chi-\chi_{0})\nabla(\psi-\psi_{0})].
\end{split}
\end{align}
Using \eqref{3.11} and \eqref{65} produces
\begin{align}\label{74}
\begin{split}
I_{2}\leq& C\int_{\frac{1}{\delta}<r<\frac{2}{\delta}}
\rho^{-4}(|\omega|^{2}+|\omega_{0}|^{2}
+r^{-2}|v-v_{0}|^{2}
+r^{-2}|\chi_{0}\psi-\psi_{0}\chi|^{2}\\
&\text{ }\text{ }\text{ }\text{ }\text{ }\text{ }\text{ }\text{ }\text{ }\text{ }
+|\psi-\psi_{0}|^{2}|\nabla(\chi-\chi_{0})|^{2}
+|\chi-\chi_{0}|^{2}|\nabla(\psi-\psi_{0})|^{2}).
\end{split}
\end{align}
The first and second terms converge to zero by the DCT and finite reduced energy of $\Psi_{0}$. Next, by \eqref{3.5} and
\eqref{64} we have $\rho^{-\frac{3}{2}}|v-v_{0}|\rightarrow 0$ as $\rho\rightarrow 0$ so that
Lemma \ref{poincare} applies, together with \eqref{3.11} to show
\begin{align}\label{75}
\begin{split}
\int_{\frac{1}{\delta}<r<\frac{2}{\delta}}
\frac{|v-v_{0}|^{2}}{r^{2}\rho^{4}}\leq &
C\int_{\frac{1}{\delta}<r<\frac{2}{\delta}}
\frac{|\nabla(v-v_{0})|^{2}}{r^{2}\rho^{2}}\\
\leq& C\int_{\frac{1}{\delta}<r<\frac{2}{\delta}}
\frac{e^{4U}}{\rho^{4}}|\omega|^{2}
+\frac{e^{4U_{0}}}{\rho^{4}}|\omega_{0}|^{2}\\
&+C\int_{\frac{1}{\delta}<r<\frac{2}{\delta}}
\frac{e^{2U}}{r^{2}\rho^{2}}(|\nabla\chi|^{2}
+|\nabla\psi|^{2})
+\frac{e^{2U_{0}}}{r^{2}\rho^{2}}(|\nabla\chi_{0}|^{2}
+|\nabla\psi_{0}|^{2})\\
&\rightarrow 0,
\end{split}
\end{align}
where the DCT and finite reduced energy were used in the last step. A similar argument holds for the fourth term on the
right-hand side of \eqref{74}, whereas the fifth and sixth terms may be directly estimated by terms in the reduced
energy of $\Psi$ and $\Psi_{0}$. It follows that $I_{2}\rightarrow 0$.

Consider the first term in the integral $I_{3}$ and write
\begin{equation}\label{76}
\nabla\chi^{1}_{\delta}=\nabla\chi_{0}
+(\chi-\chi_{0})\nabla\varphi_{\delta}^{1}
+\varphi_{\delta}^{1}\nabla(\chi-\chi_{0}),
\end{equation}
so that
\begin{equation}\label{77}
\int_{\frac{1}{\delta}<r<\frac{2}{\delta}}
\frac{e^{2U_{\delta}^{1}}}{\rho^{2}}|\nabla\chi_{\delta}^{1}|^{2}
\leq  C\int_{\frac{1}{\delta}<r<\frac{2}{\delta}}
\left(\frac{|\nabla\chi_{0}|^{2}}{\rho^{2}}
+\frac{|\chi-\chi_{0}|^{2}}{r^{2}\rho^{2}}
+\frac{|\nabla(\chi-\chi_{0})|^{2}}{\rho^{2}}\right).
\end{equation}
The first and third terms on the right-hand side may be estimated in terms of the reduced energy. The same is true of
the second term, after an application of Lemma \ref{poincare}. Since similar considerations hold for the second term in
$I_{3}$, it follows that $I_{3}\rightarrow 0$.
\end{proof}

Consider now small balls centered at the punctures $p_{n}$. Let
\begin{equation}\label{78}
F_{\delta}(\Psi)=(U,v_{\delta},\chi_{\delta},\psi_{\delta})
\end{equation}
where
\begin{equation}\label{79}
(v_{\delta},\chi_{\delta},\psi_{\delta})=
(v_{0},\chi_{0},\psi_{0})
+\varphi_{\delta}(v-v_{0},\chi-\chi_{0},\psi-\psi_{0}),
\end{equation}
so that $F_{\delta}(\Psi)=\Psi_{0}$ on $\cup_{n=1}^{N}B_{\delta}(p_{n})$.

\begin{lemma}\label{converge2}
$\lim_{\delta\rightarrow 0}\mathcal{I}(F_{\delta}(\Psi))=\mathcal{I}(\Psi).$ This also holds if $\Psi\equiv\Psi_{0}$
outside $B_{2/\delta}$.
\end{lemma}

\begin{proof}
Write
\begin{equation}\label{80}
\mathcal{I}(F_{\delta}(\Psi))
=\sum_{n=1}^{N}\left[\mathcal{I}_{r_{n}\leq\delta}(F_{\delta}(\Psi))
+\mathcal{I}_{\delta< r_{n}<2\delta}(F_{\delta}(\Psi))
+\mathcal{I}_{r_{n}\geq2\delta}(F_{\delta}(\Psi))\right],
\end{equation}
and observe that by DCT
\begin{equation}\label{81}
\sum_{n=1}^{N}\mathcal{I}_{r_{n}\geq2\delta}(F_{\delta}(\Psi))
=\sum_{n=1}^{N}\mathcal{I}_{r_{n}\geq2\delta}(\Psi)
\rightarrow \mathcal{I}(\Psi).
\end{equation}
Moreover
\begin{equation}\label{82}
\mathcal{I}_{r_{n}\leq\delta}(F_{\delta}(\Psi))=
\int_{r_{n}\leq\delta}|\nabla U|^{2}
+\frac{e^{4U}}{\rho^{4}}|\omega_{0}|^{2}
+\frac{e^{2U}}{\rho^{2}}(|\nabla\chi_{0}|^{2}
+|\nabla\psi_{0}|^{2}),
\end{equation}
where the first term on the right-hand side converges to zero again by DCT. The second and third terms may be estimated
by the reduced energy of $\Psi_{0}$ (and hence also converge to zero), since the asymptotics \eqref{18}, \eqref{19}, and
\eqref{3.14} imply that
\begin{equation}\label{82.1}
e^{U}\leq c e^{U_{0}}
\end{equation}
near each puncture.

Now write
\begin{equation}\label{83}
\mathcal{I}_{\delta< r_{n}<2\delta}(F_{\delta}(\Psi))
=\underbrace{\int_{\delta< r_{n}<2\delta}|\nabla U|^{2}}_{I_{1}}
+\underbrace{\int_{\delta< r_{n}<2\delta}
\frac{e^{4U}}{\rho^{4}}|\omega_{\delta}|^{2}}_{I_{2}}
+\underbrace{\int_{\delta< r_{n}<2\delta}
\frac{e^{2U}}{\rho^{2}}(|\nabla\chi_{\delta}|^{2}
+|\nabla\psi_{\delta}|^{2})}_{I_{3}},
\end{equation}
and observe that $I_{1}\rightarrow 0$ by DCT.
In order to estimate $I_{2}$, write
\begin{align}\label{84}
\begin{split}
\omega_{\delta}=&\varphi_{\delta}\omega
+(1-\varphi_{\delta})\omega_{0}
+(v-v_{0})\nabla\varphi_{\delta}
+(\chi_{0}\psi-\psi_{0}\chi)\nabla\varphi_{\delta}
\\
&+\varphi_{\delta}(1-\varphi_{\delta})
[(\psi-\psi_{0})\nabla(\chi-\chi_{0})
-(\chi-\chi_{0})\nabla(\psi-\psi_{0})].
\end{split}
\end{align}
Using \eqref{66} and \eqref{82.1} produces
\begin{align}\label{85}
\begin{split}
I_{2}\leq& C\int_{\delta<r_{n}<2\delta}
\left(\frac{e^{4U}}{\rho^{4}}|\omega|^{2}
+\frac{e^{4U_{0}}}{\rho^{4}}|\omega_{0}|^{2}
+\frac{e^{4U}}{r_{n}^{2}\rho^{4}}|v-v_{0}|^{2}
+\frac{e^{4U}}{r_{n}^{2}\rho^{4}}|\chi_{0}\psi-\psi_{0}\chi|^{2}\right)\\
&+C\int_{\delta<r_{n}<2\delta}
\frac{e^{4U}}{\rho^{4}}\left(|\psi-\psi_{0}|^{2}|\nabla(\chi-\chi_{0})|^{2}
+|\chi-\chi_{0}|^{2}|\nabla(\psi-\psi_{0})|^{2}\right).
\end{split}
\end{align}
The first and second terms converge to zero by the DCT and finite reduced energy of $\Psi_{0}$.

Next, assume that $p_{n}$ represents a cylindrical end, so that both $U,U_{0}\sim\log r_{n}$.
By \eqref{3.5} and \eqref{64} we have $\rho^{-\frac{3}{2}}|v-v_{0}|\rightarrow 0$ as $\rho\rightarrow 0$ so that
Lemma \ref{poincare} applies to yield
\begin{align}\label{86}
\begin{split}
\int_{\delta<r_{n}<2\delta}
\frac{e^{4U}}{r_{n}^{2}\rho^{4}}|v-v_{0}|^{2}\leq &C\int_{\delta<r_{n}<2\delta}
\frac{r_{n}^{2}}{\rho^{4}}|v-v_{0}|^{2}\\
\leq &C\int_{\delta<r_{n}<2\delta}
\frac{r_{n}^{2}}{\rho^{2}}|\nabla(v-v_{0})|^{2}\\
\leq& C\int_{\delta<r_{n}<2\delta}
\frac{e^{4U}}{\rho^{4}}|\omega|^{2}
+\frac{e^{4U_{0}}}{\rho^{4}}|\omega_{0}|^{2}\\
&+C\int_{\delta<r_{n}<2\delta}
\frac{e^{2U}}{\rho^{2}}(|\nabla\chi|^{2}
+|\nabla\psi|^{2})
+\frac{e^{2U_{0}}}{\rho^{2}}(|\nabla\chi_{0}|^{2}
+|\nabla\psi_{0}|^{2})\\
&\rightarrow 0,
\end{split}
\end{align}
where the DCT and finite reduced energy were used in the last step. A similar argument holds for the fourth term on the
right-hand side of \eqref{85}. The fifth and sixth terms may be directly estimated by terms in the reduced energy of
$\Psi$ and $\Psi_{0}$, since
\begin{equation}\label{87}
\frac{|\chi-\chi_{0}|}{\sin\theta}
+\frac{|\psi-\psi_{0}|}{\sin\theta}\leq C.
\end{equation}
To establish this, observe that by the mean value theorem
\begin{equation}\label{88}
|(\chi-\chi_{0})(r_{n},\theta)|
=\theta|\partial_{\theta}(\chi-\chi_{0})(r_{n},\theta')|
\leq Cr_{n}\theta(|\nabla\chi(r_{n},\theta')|+|\nabla\chi_{0}(r_{n},\theta')|)
\leq C\theta
\end{equation}
for some $\theta'<\theta$, where we have used \eqref{3.2} and \eqref{58.1}. Performing a similar calculation based at
$\theta=\pi$, then yields \eqref{87}, after noting that $\psi,\psi_{0}$ behave analogously to $\chi,\chi_{0}$. Lastly,
in the case that $p_{n}$ represents an asymptotically flat end, similar computations yield the desired result. It
follows that $I_{2}\rightarrow 0$.

Consider the first term in the integral $I_{3}$ and write
\begin{equation}\label{89}
\nabla\chi_{\delta}=\nabla\chi_{0}
+(\chi-\chi_{0})\nabla\varphi_{\delta}
+\varphi_{\delta}\nabla(\chi-\chi_{0}),
\end{equation}
so that
\begin{equation}\label{90}
\int_{\delta<r_{n}<2\delta}
\frac{e^{2U}}{\rho^{2}}|\nabla\chi_{\delta}|^{2}
\leq  C\int_{\delta<r_{n}<2\delta}
\left(\frac{e^{2U_{0}}}{\rho^{2}}|\nabla\chi_{0}|^{2}
+\frac{e^{2U}}{r_{n}^{2}\rho^{2}}|\chi-\chi_{0}|^{2}
+\frac{e^{2U}}{\rho^{2}}|\nabla(\chi-\chi_{0})|^{2}\right).
\end{equation}
The first and third terms on the right-hand side may be estimated in terms of the reduced energy. The same is true of
the second term, after an application of Lemma \ref{poincare} as above. Since similar considerations hold for the second
term in $I_{3}$, and it follows that $I_{3}\rightarrow 0$.
\end{proof}

Consider now cylindrical regions around the axis $\Gamma$ and away from the punctures given by
\begin{equation}\label{91}
\mathcal{C}_{\delta,\varepsilon}=
\{\rho\leq\varepsilon\}\cap
\{\delta\leq r_{n}\text{ for }n=1,\ldots,N; r\leq 2/\delta\},
\end{equation}
\begin{equation}\label{91.1}
\mathcal{W}_{\delta,\varepsilon}=
\{\varepsilon\leq \rho\leq\sqrt{\varepsilon}\}\cap
\{\delta\leq r_{n}\text{ for }n=1,\ldots,N; r\leq 2/\delta\}.
\end{equation}
Let
\begin{equation}\label{92}
G_{\varepsilon}(\Psi)=(U,v_{\varepsilon},
\chi_{\varepsilon},\psi_{\varepsilon})
\end{equation}
where
\begin{equation}\label{93}
(v_{\varepsilon},
\chi_{\varepsilon},\psi_{\varepsilon})=
(v_{0},\chi_{0},\psi_{0})
+\phi_{\varepsilon}(v-v_{0},\chi-\chi_{0},\psi-\psi_{0}),
\end{equation}
so that $G_{\varepsilon}(\Psi)=\Psi_{0}$ on $\rho\leq\varepsilon$.

\begin{lemma}\label{converge3}
Fix $\delta>0$ and suppose that $\Psi\equiv\Psi_{0}$ on $\cup_{n=1}^{N}B_{\delta}(p_{n})$, then
$\lim_{\varepsilon\rightarrow 0}\mathcal{I}(G_{\varepsilon}(\Psi))=\mathcal{I}(\Psi)$. This also holds if
$\Psi\equiv\Psi_{0}$ outside $B_{2/\delta}$.
\end{lemma}

\begin{proof}
Write
\begin{equation}\label{94}
\mathcal{I}(G_{\varepsilon}(\Psi))
=\mathcal{I}_{\mathcal{C}_{\delta,\varepsilon}}(G_{\varepsilon}(\Psi))
+\mathcal{I}_{\mathcal{W}_{\delta,\varepsilon}}(G_{\varepsilon}(\Psi))
+\mathcal{I}_{\mathbb{R}^{3}\setminus(\mathcal{C}_{\delta,\varepsilon}
\cup\mathcal{W}_{\delta,\varepsilon})}(G_{\varepsilon}(\Psi)).
\end{equation}
Since
$\Psi\equiv\Psi_{0}$ on $\cup_{n=1}^{N}B_{\delta}(p_{n})$, the DCT and finite energy of $\Psi_{0}$ imply that
\begin{equation}\label{95}
\mathcal{I}_{\mathbb{R}^{3}\setminus(\mathcal{C}_{\delta,\varepsilon}
\cup\mathcal{W}_{\delta,\varepsilon})}(G_{\varepsilon}(\Psi))
\rightarrow \mathcal{I}(\Psi).
\end{equation}
Moreover
\begin{equation}\label{96}
\mathcal{I}_{\mathcal{C}_{\delta,\varepsilon}}(G_{\varepsilon}(\Psi))=
\int_{\mathcal{C}_{\delta,\varepsilon}}|\nabla U|^{2}
+\frac{e^{4U}}{\rho^{4}}|\omega_{0}|^{2}
+\frac{e^{2U}}{\rho^{2}}(|\nabla\chi_{0}|^{2}
+|\nabla\psi_{0}|^{2}),
\end{equation}
where the first term on the right-hand side converges to zero again by DCT. The second and third terms may be estimated
by the reduced energy of $\Psi_{0}$ (and hence also converge to zero), since
\begin{equation}\label{97}
|U|+|U_{0}|\leq C\text{ }\text{ }\text{ on }\text{ }\text{ }\mathbb{R}^{3}\setminus\cup_{n=1}^{N}B_{\delta}(p_{n})
\end{equation}
by \eqref{3.11}.

Now write
\begin{equation}\label{98}
\mathcal{I}_{\mathcal{W}_{\delta,\varepsilon}}(G_{\varepsilon}(\Psi))
=\underbrace{\int_{\mathcal{W}_{\delta,\varepsilon}}|\nabla U|^{2}}_{I_{1}}
+\underbrace{\int_{\mathcal{W}_{\delta,\varepsilon}}
\frac{e^{4U}}{\rho^{4}}|\omega_{\varepsilon}|^{2}}_{I_{2}}
+\underbrace{\int_{\mathcal{W}_{\delta,\varepsilon}}
\frac{e^{2U}}{\rho^{2}}(|\nabla\chi_{\varepsilon}|^{2}
+|\nabla\psi_{\varepsilon}|^{2})}_{I_{3}},
\end{equation}
and notice that $I_{1}\rightarrow 0$ by DCT.
In order to estimate $I_{2}$, write
\begin{align}\label{99}
\begin{split}
\omega_{\varepsilon}=&\phi_{\varepsilon}\omega
+(1-\phi_{\varepsilon})\omega_{0}
+(v-v_{0})\nabla\phi_{\varepsilon}
+(\chi_{0}\psi-\psi_{0}\chi)\nabla\phi_{\varepsilon}
\\
&+\phi_{\varepsilon}(1-\phi_{\varepsilon})
[(\psi-\psi_{0})\nabla(\chi-\chi_{0})
-(\chi-\chi_{0})\nabla(\psi-\psi_{0})].
\end{split}
\end{align}
Using \eqref{67} and \eqref{97} produces
\begin{align}\label{100}
\begin{split}
I_{2}\leq& C\int_{\mathcal{W}_{\delta,\varepsilon}}
\rho^{-4}(|\omega|^{2}
+|\omega_{0}|^{2}
+(\log\varepsilon)^{-2}\rho^{-2}|v-v_{0}|^{2}
+(\log\varepsilon)^{-2}\rho^{-2}|\chi_{0}\psi-\psi_{0}\chi|^{2}\\
&\text{ }\text{ }\text{ }\text{ }\text{ }\text{ }\text{ }\text{ }\text{ }\text{ }+
|\psi-\psi_{0}|^{2}|\nabla(\chi-\chi_{0})|^{2}
+|\chi-\chi_{0}|^{2}|\nabla(\psi-\psi_{0})|^{2}).
\end{split}
\end{align}
The first and second terms converge to zero by the DCT and finite reduced energy of $\Psi_{0}$. The third term may be
directly estimated with the help of \eqref{3.5} and \eqref{64}
\begin{equation}\label{101}
\int_{\mathcal{W}_{\delta,\varepsilon}}
(\log\varepsilon)^{-2}\rho^{-6}|v-v_{0}|^{2}
\leq C\int_{\mathcal{W}_{\delta,\varepsilon}}
(\log\varepsilon)^{-2}\rho^{-2}\leq C(\log\varepsilon)^{-1}\rightarrow 0.
\end{equation}
A similar calculation holds for the fourth term. Consider now the fifth term, and use \eqref{3.5}, \eqref{59}, and
\eqref{61.2} to find
\begin{equation}\label{102}
\int_{\mathcal{W}_{\delta,\varepsilon}}
\rho^{-4}|\psi-\psi_{0}|^{2}|\nabla(\chi-\chi_{0})|^{2}
\leq C\int_{\mathcal{W}_{\delta,\varepsilon}}
\rho^{2}\leq C\varepsilon^{2}\rightarrow 0.
\end{equation}
The sixth term behaves in the same way, and hence $I_{2}\rightarrow 0$.

Consider the first term in the integral $I_{3}$ and write
\begin{equation}\label{103}
\nabla\chi_{\varepsilon}=\nabla\chi_{0}
+(\chi-\chi_{0})\nabla\phi_{\varepsilon}
+\phi_{\varepsilon}\nabla(\chi-\chi_{0}),
\end{equation}
so that
\begin{equation}\label{104}
\int_{\mathcal{W}_{\delta,\varepsilon}}
\frac{e^{2U}}{\rho^{2}}|\nabla\chi_{\varepsilon}|^{2}
\leq  C\int_{\mathcal{W}_{\delta,\varepsilon}}
\rho^{-2}(|\nabla\chi_{0}|^{2}
+(\log\varepsilon)^{-2}\rho^{-2}|\chi-\chi_{0}|^{2}
+|\nabla(\chi-\chi_{0})|^{2}).
\end{equation}
All of these terms may be estimated as above, showing that $I_{3}\rightarrow 0$.
\end{proof}

By composing the three cut and paste operations defined above, we obtain the desired replacement for $\Psi$ which
satisfies \eqref{54}. Namely, let
\begin{equation}\label{105}
\Psi_{\delta,\varepsilon}
=G_{\varepsilon}\left(F_{\delta}\left(
F_{\delta}^{1}(\Psi)\right)\right).
\end{equation}

\begin{prop}\label{prop1}
Let $\varepsilon\ll\delta\ll 1$ and suppose that $\Psi$ satisfies the hypotheses of Theorem \ref{minimum}. Then
$\Psi_{\delta,\varepsilon}$ satisfies \eqref{54} and
\begin{equation}\label{106}
\lim_{\delta\rightarrow 0}\lim_{\varepsilon\rightarrow 0}
\mathcal{I}(\Psi_{\delta,\varepsilon})=\mathcal{I}(\Psi).
\end{equation}
\end{prop}

We are now in a position to prove the main result of this section.\medskip

\noindent\textit{Proof of Theorem \ref{minimum}.} By Proposition \ref{prop1} $\Psi_{\delta,\varepsilon}$ satisfies
\eqref{54}. Thus, if $\tilde{\Psi}^{t}_{\delta,\varepsilon}$ is the geodesic connecting $\tilde{\Psi}_{0}$ to
$\tilde{\Psi}_{\delta,\varepsilon}$ as described at the beginning of this section, then
$U^{t}_{\delta,\varepsilon}=U_{0}+t(U_{\delta,\varepsilon}-U_{0})$. Following \cite{SchoenZhou} we have
\begin{equation}\label{107}
\frac{d^{2}}{dt^{2}}\mathcal{I}(\Psi^{t}_{\delta,\varepsilon})
=\underbrace{\frac{d^{2}}{dt^{2}}\mathcal{I}_{\Omega_{\delta,\varepsilon}}(\Psi^{t}_{\delta,\varepsilon})}_{I_{1}}+
\underbrace{\frac{d^{2}}{dt^{2}}\mathcal{I}_{\mathcal{A}_{\delta,\varepsilon}}(\Psi^{t}_{\delta,\varepsilon})}_{I_{2}},
\end{equation}
with
\begin{align}\label{108}
\begin{split}
I_{1}&=\frac{d^{2}}{dt^{2}}E_{\Omega_{\delta,\varepsilon}}
(\tilde{\Psi}^{t}_{\delta,\varepsilon})
+\frac{d^{2}}{dt^{2}}\int_{\partial\Omega_{\delta,\varepsilon}
\cap\partial\mathcal{A}_{\delta,\varepsilon}}
(\partial_{\nu}\log\rho)\left[2(U_{0}
+t(U_{\delta,\varepsilon}-U_{0}))-\log\rho\right]\\
&\geq 2\int_{\Omega_{\delta,\varepsilon}}
|\nabla\operatorname{dist}_{\mathbb{H}^{2}_{\mathbb{C}}}
(\Psi_{\delta,\varepsilon},\Psi_{0})|^{2}
\end{split}
\end{align}
where convexity of the harmonic energy \cite{SchoenZhou} was used in the last step, and
\begin{align}\label{109}
\begin{split}
I_{2}=&\int_{\mathcal{A}_{\delta,\varepsilon}}
2|\nabla(U_{\delta,\varepsilon}-U_{0})|^{2}
+16(U_{\delta,\varepsilon}-U_{0})^{2}\frac{e^{4[U_{0}+t(U_{\delta,\varepsilon}-U_{0})]}}{\rho^{4}}|\omega_{0}|^{2}\\
&+\int_{\mathcal{A}_{\delta,\varepsilon}}4(U_{\delta,\varepsilon}-U_{0})^{2}\frac{e^{2[U_{0}+t(U_{\delta,\varepsilon}-U_
{0})]}}{\rho^{2}}
(|\nabla\chi_{0}|^{2}+|\nabla\psi_{0}|^{2})\\
\geq &2\int_{\mathcal{A}_{\delta,\varepsilon}}
|\nabla\operatorname{dist}_{\mathbb{H}^{2}_{\mathbb{C}}}
(\Psi_{\delta,\varepsilon},\Psi_{0})|^{2}
\end{split}
\end{align}
since $\operatorname{dist}_{\mathbb{H}^{2}_{\mathbb{C}}}(\Psi_{\delta,\varepsilon},\Psi_{0})
=|U_{\delta,\varepsilon}-U_{0}|$ on $\mathcal{A}_{\delta,\varepsilon}$.

It remains to show that passing $\frac{d^{2}}{dt^{2}}$ into the integral in \eqref{109} is justified. For this it is
sufficient to show that each term on the right-hand side of the equality in \eqref{109} is uniformly integrable. There
is no issue with the first term since $U_{\delta,\varepsilon},U_{0}\in H^{1}(\mathbb{R}^{3})$. Consider now the second
and third terms, and write $\mathcal{A}_{\delta,\varepsilon}=\mathcal{C}_{\delta,\varepsilon}
\cup_{n=1}^{N}B_{\delta}(p_{n})$. Uniform integrability will follow if
$(U_{\delta,\varepsilon}-U_{0})^{2}e^{at(U_{\delta,\varepsilon}-U_{0})}$, $a=2,4$ is uniformly bounded, since then these
terms may be estimated by the reduced energy of $\Psi_{0}$. This is clearly the case on
$\mathcal{C}_{\delta,\varepsilon}$, as $U$ and $U_{0}$ are bounded on this region. On $B_{\delta}(p_{n})$,
$U_{\delta,\varepsilon}-U_{0}\sim\log r_{n}$ if $p_{n}$ represents an asymptotically flat end and
$U_{\delta,\varepsilon}-U_{0}\sim 1$ if $p_{n}$ represents an asymptotically cylindrical end. Thus, the desired
conclusion follows if $r_{n}^{at}(\log r_{n})^{2}$ is uniformly bounded, which occurs for $0<t_{0}<t\leq 1$. Since
$t_{0}>0$ is arbitrary, we conclude that \eqref{55} holds for $\Psi_{\delta,\varepsilon}$ when $t\in(0,1]$.

We now aim to verify \eqref{56} for $\Psi_{\delta,\varepsilon}$. Choose $\varepsilon_{0}<\varepsilon$,
$\delta_{0}<\delta$ and write
\begin{equation}\label{110}
\frac{d}{dt}\mathcal{I}(\Psi^{t}_{\delta,\varepsilon})
=\underbrace{\frac{d}{dt}\mathcal{I}_{\Omega_{\delta_{0},\varepsilon_{0}}}(\Psi^{t}_{\delta,\varepsilon})}_{I_{3}}+
\underbrace{\frac{d}{dt}\mathcal{I}_{\mathcal{A}_{\delta_{0},\varepsilon_{0}}}(\Psi^{t}_{\delta,\varepsilon})}_{I_{4}}.
\end{equation}
Justification for passing $\frac{d}{dt}$ into the integrals, for $t\in(0,1]$, is similar to the arguments of the
previous paragraph.
Then integrating by parts, and using the Euler-Lagrange equations \eqref{c7} satisfied by $\Psi_{0}$ together with the fact that the functionals $\mathcal{I}$ and $E$ have the same critical points, produces
\begin{equation}\label{111}
I_{3}=O(t)-\sum_{n=1}^{N}\int_{\partial B_{\delta_{0}}(p_{n})}2(U_{\delta,\varepsilon}-U_{0})\partial_{\nu}U_{0}
-\int_{\partial\mathcal{C}_{\delta_{0},\varepsilon_{0}}}
2(U_{\delta,\varepsilon}-U_{0})\partial_{\nu}U_{0}
\end{equation}
for small $t$, where $\nu$ is the unit outer normal pointing toward $M_{\text{end}}^{0}$. Next, using that
$U^{t}_{\delta,\varepsilon}=U_{0}+t(U_{\delta,\varepsilon}-U_{0})$
and $\frac{d}{dt}v^{t}_{\delta,\varepsilon}=
\frac{d}{dt}\chi^{t}_{\delta,\varepsilon}
=\frac{d}{dt}\psi^{t}_{\delta,\varepsilon}=0$ yields
\begin{align}\label{112}
\begin{split}
I_{4}=&O(t)+\int_{\mathcal{A}_{\delta_{0},\varepsilon_{0}}}
2\nabla U_{0}\cdot\nabla(U_{\delta,\varepsilon}-U_{0})
+4(U_{\delta,\varepsilon}-U_{0})\frac{e^{4[U_{0}+t(U_{\delta,\varepsilon}-U_{0})]}}{\rho^{4}}|\omega_{0}|^{2}\\
&+\int_{\mathcal{A}_{\delta_{0},\varepsilon_{0}}}2(U_{\delta,\varepsilon}-U_{0})\frac{e^{2[U_{0}+t(U_{\delta,\varepsilon
}-U_{0})]}}{\rho^{2}}
(|\nabla\chi_{0}|^{2}+|\nabla\psi_{0}|^{2}).
\end{split}
\end{align}
Observe that according to the first Euler-Lagrange equation of \eqref{c7}
\begin{align}\label{113}
\begin{split}
\int_{\partial B_{\delta_{0}}(p_{n})}(U_{\delta,\varepsilon}-U_{0})\partial_{\nu}U_{0}
=&\int_{B_{\delta_{0}}(p_{n})}
\nabla U_{0}\cdot\nabla(U_{\delta,\varepsilon}-U_{0})
+2(U_{\delta,\varepsilon}-U_{0})\frac{e^{4U_{0}}}{\rho^{4}}|\omega_{0}|^{2}\\
&+\int_{B_{\delta_{0}}(p_{n})}(U_{\delta,\varepsilon}-U_{0})\frac{e^{2U_{0}}}{\rho^{2}}
(|\nabla\chi_{0}|^{2}+|\nabla\psi_{0}|^{2}).
\end{split}
\end{align}
Note that this is justified since \eqref{3.1} implies that
\begin{equation}\label{114}
\left|\int_{\partial B_{r_{n}}(p_{n})}(U_{\delta,\varepsilon}-U_{0})\partial_{\nu}U_{0}\right|
\leq C\int_{\partial B_{r_{n}}(p_{n})}\rho^{-1}|\log r_{n}| =Cr_{n}|\log r_{n}|\rightarrow 0\text{ }\text{ }\text{ as
}\text{ }\text{ }r_{n}\rightarrow 0.
\end{equation}
It follows that
\begin{align}\label{115}
\begin{split}
\lim_{t\rightarrow 0}\frac{d}{dt}\mathcal{I}(\Psi^{t}_{\delta,\varepsilon})
=&\int_{\mathcal{C}_{\delta_{0},\varepsilon_{0}}}
2\nabla U_{0}\cdot\nabla(U_{\delta,\varepsilon}-U_{0})
+4(U_{\delta,\varepsilon}-U_{0})\frac{e^{4U_{0}}}{\rho^{4}}|\omega_{0}|^{2}\\
&-\int_{\partial\mathcal{C}_{\delta_{0},\varepsilon_{0}}}
2(U_{\delta,\varepsilon}-U_{0})\partial_{\nu}U_{0}
+\int_{\mathcal{C}_{\delta_{0},\varepsilon_{0}}}2(U_{\delta,\varepsilon}-U_{0})\frac{e^{2U_{0}}}{\rho^{2}}
(|\nabla\chi_{0}|^{2}+|\nabla\psi_{0}|^{2}).
\end{split}
\end{align}
This in fact vanishes, since \eqref{113} holds with $B_{\delta_{0}}(p_{n})$ replaced by
$\mathcal{C}_{\delta_{0},\varepsilon_{0}}$. Verification of this statement follows from
\begin{equation}\label{116}
\left|\int_{\mathcal{C}_{\delta_{0},\varepsilon_{0}}}
\nabla U_{0}\cdot\nabla(U_{\delta,\varepsilon}-U_{0})
+(U_{\delta,\varepsilon}-U_{0})\Delta U_{0}\right|
\leq\int_{\mathcal{C}_{\delta_{0},\varepsilon_{0}}}|\nabla U_{\delta,\varepsilon}|^{2}
+|\nabla U_{0}|^{2}+|\Delta U_{0}|\rightarrow 0\text{ }\text{ }\text{ as }\text{ }\text{ }
\varepsilon_{0}\rightarrow 0,
\end{equation}
which is true since $U_{\delta,\varepsilon}=U$ and $U_{0}$ have finite reduced energy and $|\Delta U_{0}|\leq
C\rho^{-2+\epsilon}$ for some $\epsilon>0$ by \eqref{3.6}. Hence \eqref{56} holds for $\Psi_{\delta,\varepsilon}$.

Now integrating \eqref{55} twice and applying the Gagliardo-Nirenberg-Sobolev inequality produces
\begin{equation}\label{117}
\mathcal{I}(\Psi_{\delta,\varepsilon})-\mathcal{I}(\Psi_{0})
\geq 2\int_{\mathbb{R}^{3}}|\nabla\operatorname{dist}_{\mathbb{H}^{2}_{\mathbb{C}}}
(\Psi_{\delta,\varepsilon},\Psi_{0})|^{2}dx
\geq C\left(\int_{\mathbb{R}^{3}}
\operatorname{dist}_{\mathbb{H}_{\mathbb{C}}^{2}}^{6}(\Psi_{\delta,\varepsilon},\Psi_{0})dx
\right)^{\frac{1}{3}}.
\end{equation}
By Proposition \ref{prop1} $\lim_{\delta\rightarrow 0}\lim_{\varepsilon\rightarrow 0}
\mathcal{I}(\Psi_{\delta,\varepsilon})=\mathcal{I}(\Psi)$, and thus in order to complete the
proof it suffices to show that the limits may be passed under the integral on the right-hand side. By the triangle
inequality, it is enough to show
\begin{equation}\label{118}
\lim_{\delta\rightarrow 0}\lim_{\varepsilon\rightarrow 0}\int_{\mathbb{R}^{3}}
\operatorname{dist}_{\mathbb{H}_{\mathbb{C}}^{2}}^{6}(\Psi_{\delta,\varepsilon},\Psi)dx
=0.
\end{equation}

As mentioned in Section \ref{sec3}, the geometry of complex hyperbolic space is invariant under
the translations $\overline{v}=v+b\chi-c\psi$, $\overline{\chi}=\chi+c$, $\overline{\psi}=\psi+b$. Then using the
triangle inequality and direct calculation produces
\begin{align}\label{119}
\begin{split}
&\operatorname{dist}_{\mathbb{H}_{\mathbb{C}}^{2}}(\Psi_{\delta,\varepsilon},\Psi)\\
\leq&\operatorname{dist}_{\mathbb{H}_{\mathbb{C}}^{2}}
((U_{\delta,\varepsilon},\overline{v}_{\delta,\varepsilon},\overline{\chi}_{\delta,\varepsilon},
\overline{\psi}_{\delta,\varepsilon}),(U,\overline{v}_{\delta,\varepsilon},\overline{\chi}_{\delta,\varepsilon},
\overline{\psi}_{\delta,\varepsilon}))
+\operatorname{dist}_{\mathbb{H}_{\mathbb{C}}^{2}}
((U,\overline{v}_{\delta,\varepsilon},\overline{\chi}_{\delta,\varepsilon},
\overline{\psi}_{\delta,\varepsilon}),(U,\overline{v},\overline{\chi}_{\delta,\varepsilon},
\overline{\psi}_{\delta,\varepsilon}))\\
&+\operatorname{dist}_{\mathbb{H}_{\mathbb{C}}^{2}}
((U,\overline{v},\overline{\chi}_{\delta,\varepsilon},
\overline{\psi}_{\delta,\varepsilon}),(U,\overline{v},\overline{\chi},
\overline{\psi}_{\delta,\varepsilon}))
+\operatorname{dist}_{\mathbb{H}_{\mathbb{C}}^{2}}
((U,\overline{v},\overline{\chi},
\overline{\psi}_{\delta,\varepsilon}),(U,\overline{v},\overline{\chi},
\overline{\psi}))\\
\leq& C\left(|U-U_{\delta,\varepsilon}|+\frac{e^{2U}}{\rho^{2}}(|\overline{v}-\overline{v}_{\delta,\varepsilon}|
+|\overline{\psi}_{\delta,\varepsilon}||\overline{\chi}-\overline{\chi}_{\delta,\varepsilon}|
+|\overline{\chi}||\overline{\chi}-\overline{\chi}_{\delta,\varepsilon}|)
+\frac{e^{U}}{\rho}(|\overline{\chi}-\overline{\chi}_{\delta,\varepsilon}|
+|\overline{\chi}-\overline{\chi}_{\delta,\varepsilon}|)\right),
\end{split}
\end{align}
where $\overline{v}_{\delta,\varepsilon}=v_{\delta,\varepsilon}+b\chi_{\delta,\varepsilon}
-c\psi_{\delta,\varepsilon}$ and similarly for $\overline{\chi}_{\delta,\varepsilon}$,
$\overline{\psi}_{\delta,\varepsilon}$. Observe that
\begin{equation}\label{120}
\int_{\mathbb{R}^{3}}|U-U_{\delta,\varepsilon}|^{6}\leq
\int_{\mathbb{R}^{3}\setminus B_{1/\delta}}|U-U_{0}|^{6}.
\end{equation}
Since $U$ and $U_{0}$ are limits in $H^{1}(\mathbb{R}^{3})$ of compactly supported functions, the Sobolev inequality
implies that $U-U_{0}\in L^{6}(\mathbb{R}^{3})$, and hence this integral converges to zero as $\delta\rightarrow 0$.
Next, we have
\begin{equation}\label{121}
\int_{\mathbb{R}^{3}}\frac{e^{12U}}{\rho^{12}}|\overline{v}-\overline{v}_{\delta,\varepsilon}|^{6}
\leq\int_{\mathbb{R}^{3}\setminus B_{1/\delta}}+\int_{\mathcal{C}_{\delta,\sqrt{\varepsilon}}}
+\sum_{n=1}^{N}\int_{B_{2\delta}(p_{n})}
\frac{e^{12U}}{\rho^{12}}|\overline{v}-\overline{v}_{0}|^{6}.
\end{equation}
By Lemma \ref{poincare} and \eqref{97}
\begin{align}\label{122}
\begin{split}
\int_{\mathbb{R}^{3}\setminus B_{1/\delta}}
\frac{e^{12U}}{\rho^{12}}|\overline{v}-\overline{v}_{0}|^{6}
\leq& C\int_{\mathbb{R}^{3}\setminus B_{1/\delta}}
\frac{1}{\rho^{12}}|\overline{v}-\overline{v}_{0}|^{6}\\
\leq& C\int_{\mathbb{R}^{3}\setminus B_{1/\delta}}
\frac{|\overline{v}-\overline{v}_{0}|^{4}}{\rho^{10}}
(|\nabla\overline{v}|^{2}+|\nabla\overline{v}_{0}|^{2})\\
\leq& C\int_{\mathbb{R}^{3}\setminus B_{1/\delta}}
\frac{e^{4U}}{\rho^{4}}
|\nabla\overline{\omega}|^{2}+\frac{e^{4U_{0}}}{\rho^{4}}|\nabla\overline{\omega}_{0}|^{2}\\
&
+C\int_{\mathbb{R}^{3}\setminus
B_{1/\delta}}\frac{e^{2U}}{\rho^{2}}(|\nabla\overline{\chi}|^{2}+|\nabla\overline{\psi}|^{2})
+\frac{e^{2U_{0}}}{\rho^{2}}(|\nabla\overline{\chi}_{0}|^{2}+|\nabla\overline{\psi}_{0}|^{2}),
\end{split}
\end{align}
since $\rho^{-6}|\overline{v}-\overline{v}_{0}|^{4}$ is bounded. This integral then converges to
zero as $\delta\rightarrow 0$, as each integrand appears in the reduced energy. Furthermore
\begin{equation}\label{123}
\int_{\mathcal{C}_{\delta,\sqrt{\varepsilon}}}\frac{e^{12U}}{\rho^{12}}|\overline{v}-\overline{v}_{0}|^{6}
\leq C\int_{\mathcal{C}_{\delta,\sqrt{\varepsilon}}}\frac{|\overline{v}-\overline{v}_{0}|^{6}}{\rho^{12}}
\leq C|\mathcal{C}_{\delta,\sqrt{\varepsilon}}|\rightarrow 0\text{ }\text{ }\text{ as }\text{ }\text{ }
\varepsilon\rightarrow 0,
\end{equation}
as \eqref{3.5} and \eqref{64} imply that $|\overline{v}-\overline{v}_{0}|\leq C\rho^{2}$ here.

Consider now an asymptotically cylindrical end represented by $p_{n}$. By choosing constants $b$ and $c$ (used to define
$\overline{v}$) appropriately in certain domains, we may assume without loss of generality that $\overline{\chi}$,
$\overline{\psi}$, $\overline{\chi}_{0}$, $\overline{\psi}_{0}$ vanish on the axis. Therefore we have that \eqref{3.2}
implies $|\overline{v}-\overline{v}|_{\Gamma}|\leq Cr_{n}^{-2}\rho^{2}$ in $B_{2\delta}(p_{n})$. Moreover \eqref{63.1}
yields
\begin{equation}\label{124}
|\overline{v}-\overline{v}_{0}|\leq C\frac{\rho^{2}}{r_{n}^{5/2}}\text{ }\text{ }\text{ on }\text{ }\text{
}B_{2\delta}(p_{n}),
\end{equation}
and similarly
\begin{equation}\label{125}
|\overline{\chi}|+|\overline{\psi}|+|\overline{\chi}_{0}|+|\overline{\psi}_{0}|\leq C\frac{\rho^{2}}{r_{n}^{3/2}}\text{
}\text{ }\text{ on }\text{ }\text{ }B_{2\delta}(p_{n}).
\end{equation}
Next, using Lemma \ref{poincare} produces
\begin{align}\label{126}
\begin{split}
\int_{B_{2\delta}(p_{n})}\frac{e^{12U}}{\rho^{12}}
|\overline{v}-\overline{v}_{0}|^{6}
\leq& C\int_{B_{2\delta}(p_{n})}\frac{r_{n}^{12}}{\rho^{12}}
|\overline{v}-\overline{v}_{0}|^{6}\\
\leq & C\int_{B_{2\delta}(p_{n})}\frac{r_{n}^{12}|\overline{v}-\overline{v}_{0}|^{4}}{\rho^{10}}
\left(|\nabla\overline{v}|^{2}+|\nabla\overline{v}_{0}|^{2}\right)\\
\leq& C\int_{B_{2\delta}(p_{n})}\frac{r_{n}^{8}|\overline{v}-\overline{v}_{0}|^{4}}{\rho^{6}}
\left(\frac{e^{4U}}{\rho^{4}}|\overline{\omega}|^{2}+
\frac{e^{4U_{0}}}{\rho^{4}}|\overline{\omega}_{0}|^{2}\right)\\
&+C\int_{B_{2\delta}(p_{n})}\frac{r_{n}^{10}|\overline{v}-\overline{v}_{0}|^{4}}{\rho^{8}}
\frac{e^{2U}}{\rho^{2}}\left(|\overline{\psi}|^{2}|\nabla\overline{\chi}|^{2}+
|\overline{\chi}|^{2}|\nabla\overline{\psi}|^{2}\right)\\
&
+C\int_{B_{2\delta}(p_{n})}\frac{r_{n}^{10}|\overline{v}-\overline{v}_{0}|^{4}}{\rho^{8}}
\frac{e^{2U_{0}}}{\rho^{2}}\left(|\overline{\psi}_{0}|^{2}|\nabla\overline{\chi}_{0}|^{2}+
|\overline{\chi}_{0}|^{2}|\nabla\overline{\psi}_{0}|^{2}\right).
\end{split}
\end{align}
From \eqref{124} and \eqref{125} it follows that
\begin{equation}\label{127}
\frac{r_{n}^{8}|\overline{v}-\overline{v}_{0}|^{4}}{\rho^{6}}
+\frac{r_{n}^{10}|\overline{v}-\overline{v}_{0}|^{4}}{\rho^{8}}\leq C,
\end{equation}
and hence \eqref{126} may be estimated by reduced energies restricted to $B_{2\delta}(p_{n})$,
which converge to zero as $\delta\rightarrow 0$. Analogous arguments hold if $p_{n}$ represents
an asymptotically flat end. We conclude that \eqref{121} converges to zero.

Similar computations show that the remaining integrals arising from the right-hand side of \eqref{119} also converge to
zero, and therefore \eqref{118} holds. \hfill\qedsymbol\medskip

\textit{Proof of Theorem \ref{thm2}.}
The asymptotic assumptions on the initial data $(g,k,E,B)$ imply that $(U,v,\chi,\psi)$ satisfy the asymptotics
\eqref{17}-\eqref{19}, \eqref{57}-\eqref{59}.
Thus Theorem \ref{minimum} applies,
and Theorem \ref{thm2} follows from \eqref{20} after setting
\begin{equation}\label{128}
\mathcal{F}(\mathcal{J}_{1},\ldots,\mathcal{J}_{N},q^{e}_{1},\ldots,q^{e}_{N},
q^{b}_{1},\ldots,q^{b}_{N})=\mathcal{M}(U_{0},v_{0},\chi_{0},\psi_{0}).
\end{equation}
\hfill\qedsymbol\medskip

Consider now Conjecture \ref{con1}, and assume that equality is achieved in \eqref{25} for initial data $(M,g,k,E,B)$
with $N>1$ black holes. If $\Psi$ denotes the associated harmonic map data, then following the proof of Theorem
\ref{thm2} yields $\Psi\equiv\Psi_{0}$. Arguments in Section \ref{sec2} suggest that $(M,g,k,E,B)$ should then give rise
to a stationary axisymmetric electrovacuum extremal black hole spacetime with disconnected horizon, and with $g$ conformally flat. It is likely that this spacetime falls into the Israel-Wilson-Perj\'{e}s class, which consists of solutions to the stationary (not necessarily axisymmetric) Einstein-Maxwell equations that are distinguished by having a conformally flat orbit space. Moreover, since the initial data set is maximal, it would then follow from \cite{CRT} that such a spacetime must be the Majumdar-Papapetrou spacetime.

\appendix

\section{Revisiting the Heuristic Arguments}\label{sec5}

The heuristic physical arguments which motivate \eqref{0} go back to Penrose's original derivation of
the Penrose inequality \cite{Penrose}. Typically in such arguments,
it is assumed that the end state of gravitational
collapse is a single Kerr-Newman black hole.
However, a more appropriate assumption for the end state is a
finite number of mutually distant Kerr-Newman black holes moving apart with
asymptotically
constant velocity. This should be the result, if for instance, two distant black
holes were initially moving away from each other sufficiently fast. We will
now describe the heuristic arguments for the mass-angular momentum-charge
inequality in this setting. It appears that this has not been previously considered in the literature.

Let $m_{i}$, $\mathcal{J}_i$,
$q_{i}$ denote the ADM masses, angular momenta, and total charges of the end state black
holes. Then the total (ADM) mass, angular momentum, and charge of the end state is
$m=\sum m_i$, $\mathcal{J}=\sum\mathcal{J}_{i}$, $q=\sum q_i$.
In a Kerr-Newman black hole these quantities satisfy the equation \cite{DainKhuriWeinsteinYamada}
\begin{equation}\label{5.1}
 m_{i}^{2}=\frac{A_{i}}{16\pi}+\frac{q_{i}^2}{2}+\frac{\pi(q_{i}^{4}+4\mathcal{J}_{i}^{2})}{A_{i}},
\end{equation}
where $A_{i}$ denotes horizon area. Moreover, as a function of $A_{i}$ (keeping $\mathcal{J}_{i}$ and $q_{i}$ fixed),
the right-hand
side is nondecreasing precisely when
\begin{equation}\label{5.2}
A_{i}\geq 4\pi\sqrt{q_{i}^{4}+4\mathcal{J}_{i}^{2}},
\end{equation}
and this inequality is always satisfied with equality only for extreme black holes.
Thus, computing the minimum value of the right-hand side of \eqref{5.1} yields
\begin{equation}\label{5.3}
 m_{i}^2\geq\frac{q_{i}^{2}+ \sqrt{q_{i}^4 + 4\mathcal{J}_{i}^2}}{2},
\end{equation}
with equality only for extreme black holes.
Let $m_{0}$, $\mathcal{J}_{0}$, $q_{0}$ denote the ADM mass, angular momentum, and total
charge of an initial state. Under appropriate hypotheses, such as axisymmetry and the existence of
a twist potential, angular momentum is conserved $\mathcal{J}_{0}=\mathcal{J}=\sum\mathcal{J}_{i}$.
Moreover, by assuming that no charged matter is present, the total charge is conserved
$q_0=q=\sum q_i$, and since gravitational waves may only carry away
positive energy $m_0 \geq m = \sum m_i$.

\begin{lemma} \label{1sqrt}
Let $a_i, b_i\in\mathbb{R}$ and let $a=\sum a_i$, $b=\sum b_i$. Then
\begin{equation}\label{5.4}
 \left(a^4+b^2\right)^{1/4} \leq \sum \left(a_i^4 + b_i^2\right)^{1/4}.
\end{equation}
\end{lemma}

\begin{proof}
Let $c_i=|b_i|^{1/2}$ and $c=\sum c_i$, then
\begin{equation}\label{5.5}
  |b|^{1/2} \leq \left(\sum|b_i|\right)^{1/2}=\left(\sum c_i^2\right)^{1/2} \leq \sum c_i = c.
\end{equation}
Hence $b^2\leq c^4$. We conclude that
\begin{equation}\label{5.6}
  \left(a^4 + b^2\right)^{1/4} \leq \left(a^4+c^4\right)^{1/4} \leq \sum \left(a_i^4+c_i^4\right)^{1/4} = \sum
\left(a_i^4 + b_i^2\right)^{1/4}.
\end{equation}
\end{proof}

\begin{lemma} \label{2sqrt}
Let $a_i, b_i\in\mathbb{R}$ and let $a=\sum a_i$, $b=\sum b_i$. Then
\begin{equation}\label{5.7}
 \sqrt{a^2 + \sqrt{a^4+b^2}} \leq \sum \sqrt{a_i^2 + \sqrt{a_i^4+b_i^2}}.
\end{equation}
\end{lemma}

\begin{proof}
By Lemma~\ref{1sqrt}
\begin{equation}\label{5.8}
  \left(a^4+b^2\right)^{1/2} \leq \left( \sum \left(a_i^4+b_i^2\right)^{1/4} \right)^2.
\end{equation}
Thus, it follows that
\begin{equation}\label{5.9}
  \sqrt{a^2 + \sqrt{a^4+b^2}} \leq \sqrt{\left( \sum a_i\right)^2 +\left( \sum \left(a_i^4+b_i^2\right)^{1/4} \right)^2}
  \leq \sum \sqrt{a_i^2 + \sqrt{a_i^4+b_i^2}}
\end{equation}
\end{proof}

Now, let $a_i=q_i$, and $b_i=2\mathcal{J}_i$, then we get
\begin{equation}\label{5.10}
  \sqrt2 m = \sqrt2 \sum m_i \geq \sum \sqrt{q_i^2 + \sqrt{q_i^4+4\mathcal{J}_i^2}} \geq  \sqrt{q^2 +
\sqrt{q^4+4\mathcal{J}^2}}.
\end{equation}
Squaring both sides yields the desired result \eqref{0}. We conclude that the heuristic arguments are
sufficiently robust to support the mass-angular momentum-charge inequality, even for spacetimes with
multiple black holes moving apart from one another at high velocities.

\section{The Extreme Kerr-Newman and Majumdar-Papapetrou Harmonic Maps}\label{sec6}

First we record formulas for the extreme Kerr-Newman harmonic map.
Recall that in Boyer-Lindquist coordinates the Kerr-Newmann metric takes the form
\begin{equation} \label{c1}
 -\frac{\Delta - a^{2}\sin^{2}\theta}{\Sigma} dt^{2}
 + \frac{2a\sin^{2}\theta}{\Sigma}\left( \widetilde{r}^{2}+a^{2}-\Delta \right)dtd\phi
+  \frac{(\widetilde{r}^{2}+a^{2})^{2} - \Delta a^{2}\sin^{2}\theta}{\Sigma} \sin^{2}\theta d\phi^{2}
+ \frac{\Sigma}{\Delta} d\widetilde{r}^{2} + \Sigma d\theta^{2}
\end{equation}
where
\begin{equation}\label{c2}
\Delta = \widetilde{r}^{2} + a^{2} +q^{2} -2m\widetilde{r},\text{ }\text{ }\text{ }\text{ }\text{ } \Sigma =
\widetilde{r}^{2} + a^{2}\cos^{2}\theta,
\end{equation}
and the electromagnetic 4-potential is given by
\begin{equation} \label{c3}
\mathbf{A} = -\frac{q_{e}\widetilde{r}}{\Sigma} \left( dt+ a\sin^{2}\theta d\phi \right)
- \frac{q_{b}\cos\theta}{\Sigma} \left( a dt + (\widetilde{r}^{2}+a^{2}) d \phi \right),
\end{equation}
The event horizon is located at the larger of the two solutions to the quadratic equation
$\Delta=0$, namely $\widetilde{r}_{+}=m+\sqrt{m^{2}-a^{2}-q^{2}}$, where the angular momentum is given by $\mathcal{J} =
ma$. For $\widetilde{r}>\widetilde{r}_{+}$ it holds that $\Delta>0$, so that a new radial coordinate may be defined by
\begin{equation}\label{c4}
r=\frac{1}{2}(\widetilde{r}-m+\sqrt{\Delta}),
\end{equation}
or rather
\begin{align} \label{c5}
\begin{split}
\widetilde{r} &= r + m + \frac{m^{2}-a^{2}-q^{2}}{4r},\text{ }\text{ }\text{ }\text{ }\text{ } m^{2} \neq a^{2}+q^{2}
\\
\widetilde{r} &= r + m, \text{ }\text{ }\text{ }\text{ }\text{ } m^{2}=a^{2}+q^{2}.
\end{split}
\end{align}
Note that the new coordinate is defined for $r>0$, and a critical point for the right-hand side of \eqref{c5}
($m^{2}\neq a^{2} + q^{2}$) occurs at the horizon, so that two isometric copies of the outer region are encoded on this
interval. The coordinates $(r,\theta,\phi)$ then form a (polar) Brill coordinate system, which is related to the
(cylindrical) Brill coordinates via the usual transformation $\rho=r\sin\theta$, $z=r\cos\theta$. Finally, the
harmonic map $(u_\KN,v_\KN,\chi_\KN,\psi_\KN):\mathbb{R}^{3}\setminus\Gamma\rightarrow \mathbb{H}^{2}_{\mathbb{C}}$,
$U_\KN=u_\KN+\log\rho$, which determines the extreme Kerr-Newman solution is given by
\begin{align}\label{c6}
\begin{split}
u_\KN&=-\frac{1}{2}\log\left[\left(\widetilde{r}^{2}+a^{2}+\frac{a^{2}\sin^{2}\theta(2m\widetilde{r}-q^{2})}
{\Sigma}\right)\sin^{2}\theta\right],\\
v_\KN&=ma\cos\theta(3-\cos^{2}\theta)-\frac{a(q^{2}\widetilde{r}-ma^{2}\sin^{2}\theta)\cos\theta\sin^{2}\theta}
{\Sigma},\\
\chi_\KN&=-\frac{qa\widetilde{r}\sin^{2}\theta}{\Sigma},\\
\psi_\KN&=\frac{q(\widetilde{r}^{2}+a^{2})\cos\theta}{\Sigma}.
\end{split}
\end{align}
The Euler-Lagrange equations satisfied by this and any other harmonic map $\Psi:\mathbb{R}^{3}\rightarrow\mathbb{H}^{2}_{\mathbb{C}}$ are given by
\begin{align}\label{c7}
\begin{split}
\Delta u-2e^{4u}|\nabla v+\chi\nabla\psi-\psi\nabla\chi|^{2}
-e^{2u}(|\nabla\chi|^{2}+|\nabla\psi|^{2})&=0,\\
\operatorname{div}\left[e^{4u}(\nabla v+\chi\nabla\psi-\psi\nabla\chi)\right]&=0,\\
\operatorname{div}(e^{2u}\nabla\chi)-2e^{4u}\nabla\chi\cdot
(\nabla v+\chi\nabla\psi-\psi\nabla\chi)&=0,\\
\operatorname{div}(e^{2u}\nabla\psi)+2e^{4u}\nabla\psi\cdot
(\nabla v+\chi\nabla\psi-\psi\nabla\chi)&=0.
\end{split}
\end{align}

Consider now the Majumdar-Papapetrou spacetime $\left(\mathbb{R}\times(\mathbb{R}^{3}
\setminus\cup_{n=1}^{N}p_{n}),ds^2\right)$ with
\begin{equation}\label{c8}
ds^2=-f^{-2}dt^2+f^{2}\delta,\text{ }\text{ }\text{ }\text{ }\text{ }
f=1+\sum_{n=1}^{N}\frac{m_{n}}{r_{n}},
\end{equation}
where $m_{n}=\sqrt{(q^{e}_{n})^{2}+(q^{b}_{n})^{2}}$ represents the mass and total eletromagnetic charge of each black
hole, $\delta$ is the Euclidean metric, and $r_{n}$ is the Euclidean distance to each puncture. Axisymmetry may be
imposed by choosing the punctures $p_{n}$ to lie on the $z$-axis. Cylindrical coordinates $(\rho,z,\phi)$ in 3-space
give rise to Brill coordinates with $U_\MP=-\log f$, and the 4-potential is given by
\begin{equation}\label{c9}
\mathbf{A}=\kappa fdt+\sqrt{1-\kappa^2}\sum_{n=1}^{N}\frac{m_{n}(z-z_{n})}{r_{n}}d\phi,\text{ }\text{ }\text{ }\text{
}\text{ }0\leq\kappa\leq 1.
\end{equation}
The constant $\kappa$ relates the electric and magnetic charges to the mass by $q_{n}^{e}=\kappa m_{n}$ and
$q_{n}^{b}=\sqrt{1-\kappa^{2}}m_{n}$. Typically the Majumdar-Papapetrou spacetime is stated without magnetic charges,
however through a duality rotation
\begin{equation}\label{c10}
E=(\cos\vartheta)\tilde{E}-(\sin\vartheta)\tilde{B},\text{ }\text{ }\text{ }\text{ }
B=(\sin\vartheta)\tilde{E}+(\cos\vartheta)\tilde{B},
\end{equation}
magnetic charge may be introduced so that $\kappa=\cos\vartheta$. Since $E=\kappa\nabla\log f$ and
$B=\sqrt{1-\kappa^{2}}\nabla\log f$, the electromagnetic potentials are obtained from \eqref{32}
\begin{equation}\label{c11}
d\chi_\MP=\kappa\rho(\partial_{z}f d\rho-\partial_{\rho}f dz),\text{ }\text{ }\text{ }\text{ }\text{
}d\psi_\MP=\sqrt{1-\kappa^{2}}\rho(\partial_{z}f d\rho-\partial_{\rho}f dz),
\end{equation}
so that
\begin{equation}\label{c12}
\chi_\MP=\kappa\sum_{n=1}^{N}\frac{m_{n}(z-z_{n})}{r_{n}},\text{ }\text{ }\text{ }\text{ }\text{
}\psi_\MP=\sqrt{1-\kappa^{2}}\sum_{n=1}^{N}\frac{m_{n}(z-z_{n})}{r_{n}}.
\end{equation}
Lastly, since this spacetime is static there is no angular momentum, and hence $v_\MP=0$. This, combined with the fact
that $\chi_\MP$ and $\psi_\MP$ are proportional leads to a harmonic map
with a 2-dimensional target that is isometric to hyperbolic space, namely
$(u_\MP,v_\MP,\chi_\MP,\psi_\MP):\mathbb{R}^{3}\setminus\Gamma\rightarrow
\mathbb{H}^{2}\subset\mathbb{H}_{\mathbb{C}}^{2}$ where $U_\MP=u_\MP+\log\rho$.

\end{document}